\ifdefined\INLONGVERSION
  \relax
\else
\newcommand{\INLONGVERSION}[1]{#1}
\newcommand{\INSHORTVERSION}[1]{}
\fi

\newcommand{\shortlong}[2]{\INSHORTVERSION{#1}\INLONGVERSION{#2}}

\shortlong{
\documentclass[sigplan,screen,acmthm=false]{acmart}
\settopmatter{}
}{
\documentclass[sigplan,screen,nonacm,acmthm=false]{acmart}
\settopmatter{printccs=false}
}

\INSHORTVERSION{%
  \usepackage{xr}
  \externaldocument[LONG-]{lics20}
}

\newcommand{\longshortref}[1]{\shortlong{\cite[\cref*{LONG-#1}]{long-version}}{\cref{#1}}}

\setcopyright{acmlicensed}
\copyrightyear{2020}
\acmYear{2020}
\acmConference[LICS '20]{Proceedings of the 35th Annual ACM/IEEE Symposium on Logic
  in Computer Science (LICS)}{July 8--11, 2020}{Saarbr\"ucken, Germany}
\acmBooktitle{Proceedings of the 35th Annual ACM/IEEE Symposium on Logic in
  Computer Science (LICS '20), July 8--11, 2020, Saarbr\"ucken, Germany}
\acmPrice{15.00}
\acmDOI{10.1145/3373718.3394759}
\acmISBN{978-1-4503-7104-9/20/07}

\usepackage{booktabs}   
\usepackage{subcaption} 

\newcommand{\andrea}[1]{}
\newcommand{\christian}[1]{}
\newcommand{\outline}[1]{}

\usepackage{macros}

\begin{document}


\title[]{Partial Univalence in $n$-truncated Type Theory}


\author{Christian Sattler}
\affiliation{
  \department{Department of Computer Science and Engineering}
  \institution{Chalmers University of Technology}
  \country{Sweden}                    
}
\email{sattler@chalmers.se}          

\author{Andrea Vezzosi}
\orcid{0000-0001-9570-9407}             
\affiliation{
  \department{Department of Computer Science}
  \institution{IT University of Copenhagen}
  \country{Denmark}                   
}
\email{avez@itu.dk}         

\begin{abstract}






It is well known that univalence is incompatible with uniqueness of
identity proofs (UIP), the axiom that all types are h-sets. This is
due to finite h-sets having non-trivial automorphisms as soon as they
are not h-propositions.

A natural question is then whether univalence restricted to
h-propositions is compatible with UIP. We answer this affirmatively by
constructing a model where types are elements of a closed
universe defined as a higher inductive type in homotopy type theory.
This universe has a path constructor for simultaneous "partial"
univalent completion, i.e., restricted to h-propositions.

More generally, we show that univalence restricted to $(n-1)$-types is
consistent with the assumption that all types are
$n$-truncated. Moreover we parametrize our construction by a suitably
well-behaved container, to abstract from a concrete choice of type
formers for the universe.

\end{abstract}


\maketitle

\section{Introduction, Motivation, and Overview}


Martin-L\"{o}f type theory \cite{MLTT} (MLTT) is a formal system useful both for dependently typed
programming and as a foundations for the development of mathematics.
It is the basis of proof assistants like Agda, Coq, Idris, Lean.

Homotopy type theory (HoTT) is a variation born out of the observation that equality proofs
in MLTT behave like paths in homotopy theory \cite{awodey-warren:2009}.
A major focus is
then to characterize the exact nature of equality for each type,
filling some gaps left underspecified by MLTT by taking inpiration
from the connection to spaces up to homotopy.

Central is Voevodsky's univalence axiom, stating that equalities of
types corresponds to equivalence of types.  From univalence other
extensionality principles follow, like function and propositional
extensionality: equality of functions corresponds to pointwise
equality, and equality of propositions corresponds to logical
equivalence.

Another important contribution is the introduction of higher inductive
types (HITs), which generalize inductive types by not only allowing elements
of the type but also equalities between them to be inductively
generated. A general example is taking the quotient of a type by a
relation, other examples are finite and countable powerset
types\cite{FinSetsCPP18,VVCPP20}, ordinal notations \cite{FXGCPP20}, syntax
of type theory up to judgemental equality \cite{AKPOPL16}, other
forms of colimits, and types of spaces for synthetic homotopy theory \cite{HoTTBook}.

HoTT also brought attention to a classification of types based on the
complexity of their equality type. We say that a type is
$(-2)$-truncated or \emph{contractible} if it is equivalent to the unit
type, we say a type $A$ is $(n+1)$-truncated when for any $x,y : A$, the
equality type $x =_A y$ is $n$-truncated. In particular $(-1)$-truncated
types are referred to as \emph{h-propositions}, and are those for
which any two elements are equal, while $0$-truncated types, whose
equality types are h-propositions, are called \emph{h-sets}.

The h-sets are the notion of set of homotopy type theory, and where
most constructions will belong when using HoTT as a foundation for set-based
mathematics or to reason about programs. Restricting oneself to types
whose equality type is an h-proposition also avoids having to
stipulate coherence conditions between different ways of proving the
same equality. Such coherence conditions might be arbitrarily complex
and not necessarily expressible within HoTT itself
\cite{KrausTYPES14}.

It would be tempting then, at least for these applications, to assume
that every type is an h-set, i.e., the uniqueness of identity proofs (UIP).
In the case of HITs, e.g. for set quotients, an
explicit equality constructor can be included to impose the desired
truncation level. However we are forced to step outside the h-sets
when considering them collectively as a type, which we call the
universe of h-sets, $\hSet$. In fact, by univalence, equalities in
$\hSet$ correspond to isomorphisms between the equated h-sets, of
which in general there are more than one.  This is often unfortunate
because of, e.g, the need to define sets by induction on a set
quotient, or the lack of a convenient type that could take the role of
a Grothendieck universe when formalizing categorical semantics in sets
or presheaves.

The counterexample of $\hSet$ however does not apply to univalence restricted to
h-propositions, i.e. proposition extensionality, since any two
proofs of logical equivalence between two propositions can be proven
equal.
Moreover results about set-truncated HITs often rely on propositional
extensionality when defining a map into h-propositions by induction.
One example is effectiveness of quotients, i.e., that equalities $[a]
=_{A/R} [b]$ between two representative of an equivalence class
correspond to proofs of relatedness $R(a,b)$.





In this paper we show, for the first time, that UIP is consistent with univalence for h-propositions,
and more generally that, for $n \geq 0$, the assumption that every
type is $n$-truncated is consistent with univalence restricted to
$(n-1)$-truncated types (\cref{consistency-result}).
We refer to this as \emph{partial univalence}.
We note that the result cannot be improved to include univalence for
$n$-truncated types, as that would imply univalence for all types, and
then we could prove that the $n$-th universe is not an $n$-type by the
main result of \cite{kraus-sattler:hierarchy-of-universes}.


We stress that existing truncated models such as the set model or groupoid model~\cite{hofmann-streicher:groupoid-model} do not model partial univalence.
Although the set model contains a univalent universe of propositions, it is not the case that the set of small sets is univalent when restricted to propositions.
Similarly, the groupoid model contains a univalent universe of sets, but the groupoid of small groupoids is not univalent when restricted to h-sets, \ie groupoids with propositional sets of morphisms.

The main challenge will thus be how to interpret the universes of
such a theory.  For a fixed collection of type formers, we show
how to overcome this in \cref{sec:simple}, where we construct a
partially univalent universe through a indexed higher inductive type.
%
%
In \cref{sec:general} we generalize the construction by a
signature of type formers given as an indexed container \cite{IxC}.
We then prove the consistency of a partially univalent $n$-truncated type
theory in \cref{sec:model}.  The proof uses a model of HoTT
capable of interpreting indexed higher inductive types to derive a
model for our theory.


The viability of MLTT as a programming language relies on the
canonicity property: every closed term is equal to one in canonical
form.  While univalence as an axiom interferes with canonicity,
cubical type theory~\cite{CCHM} has remedied this by representing
equality proofs as paths from an abstract interval type, and made
univalence no longer an axiom.  We formulate a partially univalent
$0$-truncated cubical type theory in \cref{sec:htcan}.  There we
also prove that the theory satisfies homotopy canonicity, the property that
every closed term is path equal to one in canonical form.
We believe that this result establishes an important first step towards a
computational interpretation of the theory.
In \cref{sec:conclusion} we discuss related works and conclude.

\INSHORTVERSION{
A long version of this article is available~\cite{long-version}.
It includes appendices that are used and referenced in our development.
}

\subsection{Formalization}

We have formalized the main construction, the subject of \cref{sec:simple,sec:general}, in Cubical Agda.
Separately, we have formalized~\longshortref{pushouts-along-monomorphisms} \INSHORTVERSION{(the material on pushouts of monomorphisms)} in Agda with univalence as a postulate.
The formalizations are available as supplementary material to this article.

\section{A $0$-truncated Partially Univalent Universe of $0$-types}
\label{sec:simple}

In this section we set up some preliminary definitions and notations.
We then provide a simpler case of our main technical result
\cref{general-V-truncated}, to exemplify the reasoning necessary.

When reasoning internally in type theory, we write $\equiv$ for
judgmental equality and $=$ for internal equality using the identity
type. Given a type $A$, we also write the latter as $A(-, -)$.
Given $p : a_0 =_A a_1$ and a family $B$ over $A$, the dependent equality $b_0 =_{B(p)} b_1$ of $b_0$ and $b_1$ over $p$ is shorthand for the identity type $p_*(b_0) =_{B(a_1)} b_1$, where $p_*(b_0)$ is the transport of $b_0 : B(a_0)$ to $B(a_1)$ along $p$.
Given a type $A$ and $n \geq -2$, we write $\tlevel{n}(A)$ for the type that $A$ is $n$-truncated.
For $n = -2, -1, 0$, we have the usual special cases $\iscontr(A), \isprop(A), \isset(A)$ of $A$ being contractible, propositional, and a set, respectively.
All of these types are propositions.

We recall the notion of a univalent family.
\begin{definition}[Univalence] \label{univalence}
A family $Y$ over a type $X$ is \emph{univalent} if the canonical map $X(x_0, x_1) \to Y(x_0) \simeq Y(x_1)$ is an equivalence for all $x_0, x_1 : X$.
\end{definition}

We relativize this notion with respect to a property $\Pred$ on types.
This is supposed to be an extensional property, in the sense that it should depend only on $Y(x)$, not on the ``code'' $x : X$.
To make this precise, we let $Y$ be a valued in a universe $\UU$ and express $\Pred$ as a propositional family over $\UU$.

We use the word universe in a rather weak sense: until we add closure under some type formers, it can refer to an arbitrary type family.
Notationwise, universes are distinguished in that we leave the decoding function from elements of $\UU$ to types implicit.

\begin{definition}[Partial univalence] \label{partial-univalence}
Let $X$ be a type and $Y \co X \to \UU$ for some universe $\UU$.
Let $\Pred$ be a propositional family over $\UU$.
We say that $(X, Y)$ is \emph{$\Pred$-univalent} if the restriction of the family $Y$ to the subtype $\Sigma_{x:X} \Pred(Y(x))$ is univalent.
\end{definition}

We also say that $X$ is \emph{partially univalent} or \emph{univalent for $\Pred$}, leaving $Y$ implicit.
For $\Pred \equiv \tlevel{n}$, we say that $X$ is \emph{univalent for $n$-types}.
In particular, for $\Pred \equiv \isprop$, we say that $X$ is \emph{univalent for propositions}.

Of particular importance is the case where $Y$ is the identity function.
In that case, we say that \emph{the universe $\UU$ is $\Pred$-univalent}.

\begin{lemma} \label{partial-univalence-via-embedding}
Let the universe $\UU$ be $\Pred$-univalent.
Then $Y \co X \to \UU$ is $\Pred$-univalent exactly if its restriction to $x : X$ with $\Pred(Y(x))$ is an embedding.
\end{lemma}

\begin{proof}
Let $x_1, x_2 : X$ with $\Pred(Y(x_1))$ and $\Pred(Y(x_2))$.
Consider the commuting diagram
\[
\xymatrix@!C@C-1.5cm{
  X(x_1, x_2)
  \ar[rr]
  \ar[dr]
&&
  \UU(Y(x_1), Y(x_2))
  \ar[dl]
\\&
  Y(x_1) \simeq Y(x_2)
\rlap{.}}
\]
Since $\UU$ is $\Pred$-univalent, the right map is an equivalence.
By 2-out-of-3, the left map is invertible exactly if the top map is invertible.
Quantifying over $x_1, x_2$, we obtain the claim.
\end{proof}

\subsection{Partially Univalent Type $V$ of (Small) Sets}

Let now $\UU^{\leq 0}$ be a univalent universe of sets, meaning its elements decode to $0$-truncated types.
We wish to define a ``closed'' $0$-truncated universe $V$ with a decoding function $\El_V : V \to \UU^{\leq 0}$ that is univalent for propositions in the sense of \cref{partial-univalence}.
We illustrate the essential features of our construction by requiring that:
\begin{itemize}
\item $V$ contains codes for a fixed family $N : M \to \UU^{\leq 0}$ of elements of $\UU^{\leq 0}$ where $M$ is a set.
\item $V$ is closed under $\Pi$-types (assuming that $\UU^{\leq 0}$ is),
\end{itemize}
The former family can for example include codes for the empty type or the type of Booleans.

Closed universes are typically defined by induction\hyp{}recursion, simultaneously defining the type $V$ and the function $\El_V : V \to \UU^{\leq 0}$.
To model the above closure conditions, one takes:
\begin{itemize}
\item
given $m : M$, a constructor $\overline{N}(m) : V$ and a clause $\El_V(\overline{N}(m)) \equiv N(m)$.
\item
given $\overline{A} : V$ and $\overline{B} : \El_V(\overline{A}) \to V$, a constructor $\overline{\Pi}(\overline{A}, \overline{B}) : V$ and a clause
\[
\El_V(\overline{\Pi}(\overline{A}, \overline{B})) \equiv \textstyle{\prd{a : \El_V(\overline{A})} \El_V(\overline{B}(a))}
,\]
\end{itemize}

In order to make $V$ univalent for propositions, one could imagine turning this into a \emph{higher inductive\hyp{}recursive} definition.
Given $\overline{A}_i : V$ with $\El_V(\overline{A}_i)$ a proposition for $i \in \{0, 1\}$ and an equality $p : \UU^{\leq 0}(\El_V(\overline{A}_0), \El_V(\overline{A}_1))$, one would add a path constructor $\ua(e) : V(\overline{A}_0, \overline{A}_1)$ with a clause giving an identification of the action of $\El_V$ on the path $\ua(e)$ with $p$.%
\footnote{
In this particular case, the clause of $\El_V$ for the path constructor amounts to nothing as it is an identification in a propositional type.
}
This is the right idea, but there are problems.
\begin{itemize}
\item
While syntax and semantics of higher inductive types have been analyzed to a certain extent~\cite{lumsdaine-shulman:hits-semantics,kaposi-kovacs:hits-syntax} this analysis does not yet extend to the case of induction\hyp{}recursion.
As such, the rules for higher inductive\hyp{}recursive types have not yet been established and none of the known models of homotopy type theory have been shown to admit them.
\item
Induction\hyp{}recursion is known to increase the proof-theoretic strength of the type theory over just (indexed) induction.
Thus, we do not wish to assume it in our ambient type theory.
\end{itemize}

Given a type $I$, recall that types $X$ with a map $X \to I$ are equivalent to families over $I$: in the forward direction, one takes fibers; in the backward direction, one takes the total type.
Exploiting this correspondence, the above inductive\hyp{}recursive definition of $V$ (without path constructor for partial univalence) can be turned into an \emph{indexed inductive} definition of a family $\inV$ over $\UU^{\leq 0}$.
The translation of the constructors for $\Pi$-types and $M$ is given in \cref{simple-in-V:pi,simple-in-V:M} of \cref{simple-in-V}.
The path constructor for partial univalence corresponds to the following: given propositions $A_i : \UU^{\leq 0}$ with $w_{A_i} : \inV(A_i)$ for $i \in \braces{0, 1}$ and an equality $p : \UU^{\leq 0}(A_0, A_1)$, we have a path $\ua(e)$ in the family $\inV$ betwen $w_{A_0}$ and $w_{A_1}$ over $p$.
We contract the path $p$ with one of its endpoints and arrive at the definition below.

\begin{definition} \label{simple-in-V}
The family $\inV$ over $\UU^{\leq 0}$ is defined as the following higher indexed inductive type:
\begin{enumerate}[label=(\roman*)]
\item \label{simple-in-V:M}
given $m : M$, a constructor $w_N(m) : \inV(N(m))$,
\item \label{simple-in-V:pi}
given $w_A : \inV(A)$ and $w_B(a) : \inV(B(a))$ for $a : A$ (with implicit $A : \UU^{\leq 0}$ and $B : A \to \UU^{\leq 0}$), a constructor $w_\Pi(w_A, w_B) : \inV(\prd{a : A} B(a))$,
\item \label{simple-in-V:univalence}
given a proposition $X : \UU^{\leq 0}$ with $w_0, w_1 : \inV(X)$, a path constructor $\ua(w_0, w_1) : w_0 = w_1$.
\end{enumerate}
\end{definition}

We recover $V$ as the total type $V = \sm{X : \UU^{\leq 0}} \inV(X)$, with $\El_V$ given by the first projection.

We may regard $\inV$ as a \emph{conditionally} or \emph{partially propositionally truncated} indexed inductive type (see \longshortref{appendix-on-W-type-stuff}).
In this form, it becomes clear that the constructor $\ua$ indeed suffices for partial univalence and does not introduce coherence problems: it exactly enforces that the restriction of the family $\inV$ to elements decoding to propositions is valued in propositions.

\begin{lemma} \label{simple-V-partially-univalent}
The type $V$ with $\El_V : V \to \UU^{\leq 0}$ is univalent for propositions.
\end{lemma}

\begin{proof}
Using \cref{partial-univalence-via-embedding}, we have to show that the restriction of $\El_V \co V \to \UU^{\leq 0}$ to $\overline{X} : V$ with $\El_V(\overline{X})$ a proposition is an embedding.
Unfolding to $\inV$, this says that $\inV(X)$ is propositional for $X : \UU^{\leq 0}$ a proposition.
This is exactly enforced by the path constructor~\cref{simple-in-V:univalence} in \cref{simple-in-V}.
\end{proof}

It remains to show that $V$ is $0$-truncated.
For this, we adapt the encode-decode method to characterize the dependent equalities in $\inV$ over an equality in $\UU^{\leq 0}$.

\subsection{Dependent Equalities in $\inV$}

In the following, we make use of (homotopy) pushouts.
Recall~\cite[Section~6.8]{HoTTBook} that the \emph{pushout} of a span $f \co A \to B$ and $g \co A \to C$ of types is the (non-recursive) higher inductive type $B +_A C$ with points constructors $\inl(b) : B +_A C$ for $b : B$ and $\inr(c) : B +_A C$ for $c : C$ and path constructor $\glue(a) : (B +_A C)(\inl(f(b)), \inr(g(c)))$ for $a : A$.

As in~\cite[Lecture~13]{egbert:hott-intro}, we do not require judgmental $\beta$\hyp{}laws even for point constructors.
Thus, pushout types are simply a particular choice of pushout squares, (homotopy) initial cocones under the span $B \leftarrow A \rightarrow C$.
We refer to \longshortref{pushouts-along-monomorphisms} for some key properties of pushouts used in our development.

An important special case is the \emph{join} $X \join Y$ of types $X$ and $Y$, the pushout of $X$ and $Y$ under $X \times Y$.
For a proposition $\Pred$, the operation $\Pred \join -$ is also known as the \emph{closed modality} associated with $\Pred$~\cite[Example~1.8]{rijke-shulman-bas:modalities}.
We will only use the join in this form.
Recall that $X \join Y$ is contractible if $X$ or $Y$ is contractible.
In particular, $\Pred \join X$ is contractible if $\Pred$ holds.

\andrea{todo: announce motivation/inspiration for the conditional contractiblity}
\begin{problem} \label{simple-in-V-Eq}
Given an equality $p \co \UU^{\leq 0}(X_0, X_1)$ and $w_i : \inV(X_i)$ for $i \in \braces{0, 1}$, we wish to define: 
\begin{itemize}
\item
a type $\Eq_\inV(p, w_0, w_1)$ of \emph{codes of equalities} over $p$ between $w_0$ and $w_1$
\item
such that $\Eq_\inV(p, w_0, w_1)$ is contractible if $X_0$ (or equivalently $X_1$) is a proposition.
\end{itemize}
\end{problem}

\begin{proof}[Construction]
By univalence for contractible types, the type of contractible types is contractible.
Thus, the goal is contractible if $X_0$ or $X_1$ is a proposition.

We perform double induction, first on $w_0 : \inV(X_0)$ and then on $w_1 : \inV(X_1)$.
In all path constructor cases, we know that $X_0$ or $X_1$ is a proposition.
By the above, the goal becomes an equality in a contractible type, so there is nothing to show.

In all point constructor cases, we define
\[
\Eq_\inV(u, w_0, w_1) \defeq \isprop(X_0) \join E
\]
where $E$ is an abbreviation for an expression that varies depending on the case.
Note that this makes $\Eq_\inV(u, w_0, w_1)$ contractible when $X_0$ is a proposition.
The expression $E$ codes structural equality of the top-level constructors.
\begin{itemize}
\item
For $w_0 \equiv w_N(m_0)$ and $w_1 \equiv w_N(m_1)$, we let $E$ consist of pairs $(p_m, c)$ where $p_m : M(m_0, m_1)$ and $c$ is a proof that $p : \UU^{\leq 0}(N(m_0), N(m_1))$ is equal to the action of $M$ on $p_m$.
\item
For $w_i \equiv w_\Pi(w_{A_i}, w_{B_i})$ with $w_{A_i} : \inV(A_i)$ and $w_{B_i}(a_i) : \inV(B_i(a_i))$ for $a_i : A_i$, all for $i \in \braces{0, 1}$, we let $E$ consist of tuples $(p_A, e_A, p_B, e_B, c)$ where:
\begin{itemize}
\item
$p_A : \UU^{\leq 0}(A_0, A_1)$ with $e_A : \Eq_\inV(p_A, w_{A_0}, w_{A_1})$,
\item
for $a_0 : A_0$, $a_1 : A_1$, and a dependent equality $p_a$ over $p_A$ between $a_0$ and $a_1$, we have $p_B : \UU^{\leq 0}(B_0(a_0), B_1(a_1))$ with $e_B : \Eq_\inV(p_B, w_{B_0}(a_0), w_{B_1}(a_1))$.
\item
$c$ witnesses that
\[
p : \UU^{\leq 0}(\textstyle{\prd{a_0 : A_0} B_0(a_0)}, \textstyle{\prd{a_1 : A_1} B_1(a_1)})
\]
is equal to the action of the type forming operation $\Pi$ on $p_A$ and $p_B$.
\end{itemize}
\item
In the remaining ``mixed'' cases, we let $E$ be empty.
\qedhere
\end{itemize}
\end{proof}

\begin{proposition} \label{simple-equality-in-V}
Given $p : \UU^{\leq 0}(X_0, X_1)$ and $w_i : \inV(X_i)$ for $i \in \braces{0, 1}$, there is an equivalence between dependent equalities in $\inV$ over $p$ between $w_0$ and $w_1$ and $\Eq_\inV(p, w_0, w_1)$.
\end{proposition}

\begin{proof}
For the purpose of this proof, it will be convenient to work with a different, but equivalent definition of the expression $E$ in the construction of \cref{simple-in-V-Eq} in the case $w_i \equiv w_\Pi(w_{A_i}, w_{B_i})$ with $w_{A_i} : \inV(A_i)$ and $w_{B_i}(a_i) : \inV(B_i(a_i))$ for $a_i : A_i$, all for $i \in \braces{0, 1}$.
Namely, we let $E$ consist of pairs $(q, r)$ as follows.
\begin{itemize}
\item
The component $q$ is an equality $(A_0, B_0) = (A_1, B_1)$ in the dependent sum $\sm{A : \UU^{\leq 0}} A \to \UU^{\leq 0}$.
\item
Inducting on the equality $q$, we may suppose $A \defeq A_0 \equiv A_1$ and $B \defeq B_0 \equiv B_1$.
The component $r$ is then a triple $(e_A, e_B, c)$ where
\begin{itemize}
\item $e_A : \Eq_\inV(\refl_A, w_{A_0}, w_{A_1})$,
\item $e_B(a) : \Eq_\inV(\refl_{B(a)}, w_{B_0}(a), w_{B_1}(a))$ for $a : A$,
\item $c : p = \refl$.
\end{itemize}
\end{itemize}
The equivalence between this choice of $E$ and the previous one is a staightforward consequence of structural equivalences, splitting up the equality $q$ into components for $A$ and $B$ and distributing them over the components of $r$.

We follow the encode-decode method as described in \longshortref{encode-decode}.
To define
\[
\xymatrix@C+1.2cm{
  w_0 =_{\inV(p)} w_1
  \ar[r]^-{\encode_{p,w_0,w_1}}
&
  \Eq_\inV(p, w_0, w_1)
\rlap{,}}
\]
we use equality induction on $p$ and the argument, reducing the goal to $\encode'(w) : \Eq_\inV(\refl_X, w, w)$ for $w : \inV(X)$.
We induct on $x : \inV(X)$.
\begin{itemize}
\item
For $w \equiv w_N(m) : \inV(N(m))$, we take
\[
\encode'(x) \equiv \inr(\refl_m, \refl)
.\]
\item
For $w \equiv w_\Pi(w_A, w_B) : \inV(\prd{a:A} B(a))$, we take
\[
\begin{split}
& \encode'(x) \equiv \inr
\\
& \quad (\refl_{(A, B)}, (\encode'(w_A), \lam{a} \encode'(w_B(a)), \refl))
.\end{split}
\]
\item
In the path constructor case, we have that $X$ is a proposition.
Then the goal is a dependent equality in a contractible type.
\end{itemize}

We now show $\encode_{p,w_0,w_1}^{-1}(e)$ for $e : \Eq_\inV(p, w_0, w_1)$.
We use double induction on $w_0$ and $w_1$.
In all path constructor cases, we know that $X_0$ or $X_1$ is a proposition (hence both are).
Thus, both source and target of $\encode_{p,w_0,w_1}$ are contractible, so the goal becomes contractible.
In all point constructor cases, we have $e : \isprop(X_0) \join E$ where $E$ depends on the particular case.
We induct on $e$.
In the case for $\inl$ or $\glue$, we have $\isprop(X_0)$, and the goal becomes contractible. 
For $e \equiv \inr(z)$, we proceed with $z$ according to the point constructor case for $w_0$ and $w_1$.
\begin{itemize}
\item
For $w_0 \equiv w_N(m_0)$ and $w_1 \equiv w_N(m_1)$, we have $z = (p_m, c)$.
By equality induction on $p_m : M(m_0, m_1)$, we may suppose $m \defeq m_0 \equiv m_1$ and $p_m \equiv \refl_m$.
By equality induction on $c$, we may then suppose $p \equiv \refl$ and $c \equiv \refl$.
We have
\begin{align*}
e
&\equiv
\inr(\refl_m, \refl)
\\&\equiv
\encode'(w_N(m))
\\&\equiv
\encode_{\refl_{N(m)}, w_N(m), w_N(m)}(\refl)
,\end{align*}
showing $\encode_{p,w_0,w_1}^{-1}(e)$.
\item
For $w_i \equiv w_\Pi(w_{A_i}, w_{B_i})$ with $w_{A_i} : \inV(A_i)$ and $w_{B_i}(a_i) : \inV(B_i(a_i))$ for $a_i : A_i$, all for $i \in \braces{0, 1}$, we have $z = (q, r)$ as described at the beginning of this proof.
By equality induction on $q : (A_0, B_0) = (A_1, B_1)$, we may suppose $A \defeq A_0 \equiv A_1$, $B \defeq B_0 \equiv B_1$, and $q \defeq \refl$.
Then $r = (e_A, e_B, c)$.
By equality induction $c$, we may suppose that $p \equiv \refl$ and $c \equiv \refl$.
By induction hypothesis, we have
\begin{align*}
&\encode_{\refl_A,w_{A_0},w_{A_1}}^{-1}(e_A)
,\\
&\text{$\encode_{\refl_{B(a)},w_{B_0}(a),w_{B_1}(a)}^{-1}(e_B(a))$ for $a : A$.}
\end{align*}
By equality induction and function extensionality, we may thus suppose that
\begin{align*}
e_A &\equiv \encode_{\refl_A,w_{A_0},w_{A_1}}(q_A)
,\\
e_B &\equiv \lam{a} \encode_{\refl_{B(a)},w_{B_0}(a),w_{B_1}(a)}(q_B(a))
\end{align*}
for some $q_A : w_{A_0} = w_{A_1}$ and $q_B(a) : w_{B_0}(a) = w_{B_1}(a)$ for $a : A$.
By equality induction on $q_A$ and $q_B$ (after using function extensionality), we may suppose that $w_A \defeq w_{A_0} \equiv w_{A_1}$ and $q_A \equiv \refl$ as well as $w_B \defeq w_{B_0} \equiv w_{B_1}$ and $q_B \equiv \lam{a} \refl$.
Now we have
\begin{align*}
e
&\equiv
\inr(\refl_{(A,B)}, (\encode'(w_A), \lam{a} \encode'(w_B(a)), \refl))   
\\&\equiv
\encode'(w_\Pi(w_A, w_B))
\\&\equiv
\encode_{\refl_{\prd{a:A} B(a)}, w_\Pi(w_A, w_B), w_\Pi(w_A, w_B)}(\refl)
,\end{align*}
showing $\encode_{p,w_0,w_1}^{-1}(e)$.
\item
In all ``mixed'' cases, we have $z : 0$.
\qedhere
\end{itemize}
\end{proof}

For readers concerned with the length of the above argument, we note the following.
In \cref{sec:general}, we will motivate abstraction that will allow us to reorganize the above argument into smaller, more general pieces.

\subsection{$V$ is a set}

From our characterization of dependent equality in $\inV$, we obtain a corresponding characterization of equality in $V$.
Given $\overline{X}_i \equiv (X_i, w_i) : V$ for $i \in \braces{0, 1}$, we define
\[
\Eq_V(\overline{X}_0, \overline{X}_1) \defeq \sm{p : \UU^{\leq 0}(X_0, X_1)} \Eq_\inV(p, w_0, w_1)
.\]

\begin{corollary} \label{simple-equality-V}
For $\overline{X}_0, \overline{X}_1 : V$, we have
\[
V(\overline{X}_0, \overline{X}_1) \simeq \Eq_V(\overline{X}_0, \overline{X}_1)
.\]
\end{corollary}

\begin{proof}
Equality types in the dependent sum $V \equiv \sm{X : \UU^{\leq 0}} \inV(X)$ are dependent sums of an equality in $\UU^{\leq 0}$ and a dependent equality over it.
Thus, the claim is a consequence of \cref{simple-equality-in-V}.
\end{proof}

\begin{proposition} \label{simple-V-is-set}
The type $V$ is $0$-truncated.
\end{proposition}

\begin{proof}
Given $\overline{X}_i \equiv (X_i, w_i) : V$ for $i \in \braces{0, 1}$, we wish to show $V(\overline{X}_0, \overline{X}_1)$ propositional.
By \cref{simple-equality-V}, this amounts to showing $\Eq_V(\overline{X}_0, \overline{X}_1)$ propositional.
This we show by double induction, first on $w_0 : \inV(X_0)$ and then on $w_1 : \inV(X_1)$.
Since the goal is propositional, there is nothing to show in the path constructor cases.

In all point constructor cases, we have
\begin{equation} \label{simple-V-is-set:0}
\begin{gathered}
\Eq_V(\overline{X}_0, \overline{X}_1) = \sm{p : \UU^{\leq 0}(X_0, X_1)} \isprop(X_0) \join E(p)
\end{gathered}
\end{equation}
where $E(p)$ is as in the construction for \cref{simple-in-V-Eq}, abbreviating an expression depending on the point constructor case (for clarity, we have made the dependency on $p$ explicit).
By definition, this join forms a pushout square
\begin{equation} \label{simple-V-is-set:1}
\begin{gathered}
\xymatrix{
  \isprop(X_0) \times E(p)
  \ar[r]
  \ar[d]
&
  \isprop(X_0)
  \ar[d]
\\
  E(p)
  \ar[r]
&
  \isprop(X_0) \join E(p)
  \fancypullback{[u]}{[l(0.5)]}
\rlap{.}}
\end{gathered}
\end{equation}
The dependent sum over a fixed type preserves pushout squares in its remaining argument (abstractly, because it is a higher functor left adjoint to weakening).
From~\cref{simple-V-is-set:0,simple-V-is-set:1}, we thus obtain the following pushout square:
\begin{equation} \label{simple-V-is-set:2}
\hspace{-1cm} 
\begin{gathered}
\xymatrix@C-0.5cm{
  \sm{p : \UU^{\leq 0}(X_0, X_1)} \isprop(X_0) \times E(p)
  \ar[r]
  \ar@{^(->}[d]
&
  \sm{p : \UU^{\leq 0}(X_0, X_1)} \isprop(X_0)
  \ar[d]
\\
  \sm{p : \UU^{\leq 0}(X_0, X_1)} E(p)
  \ar[r]
&
  \Eq_V(\overline{X}_0, \overline{X}_1)
  \fancypullback{[u]}{[l(0.5)]}
\rlap{.}}
\end{gathered}
\hspace{-1cm} 
\end{equation}
Since $\UU^{\leq 0}$ is univalent, the type $\UU^{\leq 0}(X_0, X_1)$ is propositional if $\isprop(X_0)$.
From this, we see that the span in~\cref{simple-V-is-set:2} is a (homotopy) product span.
By invariance of pushouts under equivalence, it follows that $\Eq_V(\overline{X}_0, \overline{X}_1)$ is equivalent to the join
\begin{equation} \label{simple-V-is-set:3}
\parens*{\UU^{\leq 0}(X_0, X_1) \times \isprop(X_0)} \join \parens[\Big]{\sm{p : \UU^{\leq 0}(X_0, X_1)} E(p)}
.\end{equation}
We now apply \longshortref{join-prop-level}: to show that this join is propositional, it suffices to show that each of its factors is propositional.%
\footnote{Alternatively, we could appeal to the fact that closed modalities are left exact, hence preserve truncation levels~\cite{rijke-shulman-bas:modalities}.}
Since $\isprop(X_0)$ is propositional and $\UU^{\leq 0}(X_0, X_1)$ is propositional if $\isprop(X_0)$, their product is propositional.

It remains to show that $T \defeq \sm{p : \UU^{\leq 0}(X_0, X_1)} E(p)$ is a proposition.
For this, we argue according to the current point constructor case, recalling the corresponding definition of $E(p)$ from \cref{simple-in-V-Eq}.
\begin{itemize}
\item
In the case $w_0 \equiv w_N(m_0)$ and $w_1 \equiv w_N(m_1)$, we have
\begin{align*}
T
&\equiv
\sm{p : \UU^{\leq 0}(X_0, X_1)}{p_m : M(m_0, m_1)} (p = \ap_N(p_m))
\\&\simeq
M(m_0, m_1)
,\end{align*}
a proposition since $M$ was assumed a set.
\item
In the case $w_i \equiv w_\Pi(w_{A_i}, w_{B_i})$ with $w_{A_i} : \inV(A_i)$ and $w_{B_i}(a_i) : \inV(B_i(a_i))$ for $a_i : A_i$ for $i \in \braces{0, 1}$, recall that $T$ consists of tuples $(p, p_A, e_A, p_B, e_B, c)$ with types as in the construction of \cref{simple-in-V-Eq}.
We contract the equality $c$ with its endpoint $p$.
What remains is equivalent to the dependent sum of:
\begin{itemize}
\item
$(p_A, e_A) : \Eq_V(\overline{A}_0, \overline{A}_1)$ where $\overline{A}_i \defeq (A_i, w_i)$,
\item
for $a_0 : A_0$, $a_1 : A_1$, and a dependent equality $p_a$ over $p_A$ between $a_0$ and $a_1$:
\[
(p_B, e_B) : \Eq_V(\overline{B}_0(a_0), \overline{B}_1(a_1))
\]
where $\overline{B}_i(a_i) \defeq (B_i(a_i), w_i(a_i))$.
\end{itemize}
Both $\Eq_V(\overline{A}_0, \overline{A}_1)$ and $\Eq_V(\overline{B}_0(a_0), \overline{B}_1(a_1))$ (the latter for all $a_0, a_1$) are $(n-1)$-truncated by induction hypothesis.
The claim now follows by closure of $(n-1)$-truncated types under dependent sums and dependent products (with arbitrary domain).
\item
In the remaining, ``mixed'' cases, we have
\[
T
\equiv
\sm{p : \UU^{\leq 0}(X_0, X_1)} \bot
\simeq
\bot
,\]
which is propositional.
\qedhere
\end{itemize}
\end{proof}

\section{$n$-truncated Partially Univalent Universes of $n$-types}
\label{sec:general}
To obtain the main result of this paper, we need to generalize the
constructions of the previous section to a partially univalent $n$-truncated universe of $n$-types rather
than sets, and to a universe closed under more type formers. For the
sake of generality, we will also build a universe that is
$\Pred$-univalent for an arbitrary proposition $\Pred$, although the
$n$-truncatedness result will need $\Pred(X)$ to imply
$\tlevel{(n-1)}(X)$.

This section only makes use of univalence for $(n-1)$-types.

\subsection{Indexed Containers and Preservation of Truncation}
To abstract from a particular choice of type formers, we will
parametrize our universe by a signature of them represented by an
indexed container \cite{IxC}, as is done for indexed W-types.
\andrea{comment about our universe being a W-type with simultaneous partial univalence completion?}

We recall here the precise definitions of indexed container and its extension that we will use in the rest of the paper.
\begin{definition}[Indexed container] \label{ixcont}
  Given a type $I$, an \emph{$I$-indexed container} is a pair $(S,\Pos)$ of a type family $S$ over $I$ and a type family $\Pos$ over
  $\sm{i : I} S(i) \times I$.
\end{definition}
\begin{definition}[Extension of a container] \label{extension}
  Let $(S,\Pos)$ be an $I$-indexed container.
  Its extension $\cExt_{S,\Pos}$ takes a family $F$ over $I$ and produces another:

  $\cExt_{S,\Pos}(F,i) = \sm{s : S(i)} \prd{j : I} \Pos(i,s,j) \to F(j)$
\end{definition}
Given $(s, t) : \cExt_{S,\Pos}(F,i)$ we will write $t_Y(p)$ for $t(Y)(p)$.

In the universe construction we will use a $\UU^{\leq n}$-indexed
container, here we demonstrate by example that they not only cover the
type formers considered in \cref{sec:simple}, but also ones
with a more complex signature like (truncated) pushouts.
\begin{example}[Nullary type formers] \label{ex:nullary}
  Given a fixed family of types $N : M \to \UU^{\leq n}$ we define a container with empty positions:
  \[
  \begin{array}{l}
    S(X) = \sm{m : M} (X = N(m))\\
    \Pos(\_,\_,\_) = \bot\\
  \end{array}
  \]
\end{example}
\begin{example}[$\Pi$-types]\label{ex:pi-cont}
  The signature for $\Pi$-types can be represented by a $\UU^{\leq n}$-indexed container where both $S$ and $\Pos$ are given by
  indexed inductive types with constructors:
  \begin{itemize}
  \item given $A :  \UU^{\leq n}$ and $B : A \to \UU^{\leq n}$ a constructor $\pi(A,B) : S(\prd{x : A} B(x))$.
  \end{itemize}
  and with $s \equiv \pi(A,B)$ and $X \equiv \prd{a : A} B(a)$:
  \begin{itemize}
  \item a constructor $\mathsf{pos}_A : \Pos(X,s,A)$
  \item given $a : A$ a constructor $\mathsf{pos}_B : \Pos(X,s,B(a))$
  \end{itemize}
\end{example}
\begin{example}[Truncated pushouts] \label{ex:po-cont}
  Pushouts truncated to be $n$-types can also be represented as a $\UU^{\leq n}$-indexed container:
  \begin{itemize}
  \item given $A_i : \UU^{\leq n}$ for $i \in \{0,1,2\}$ and $f : A_0 \to A_1$ and $g : A_0 \to A_2$ a constructor $\mathsf{po}(f,g) : S(A_1 +^n_{A_0} A_2)$,
  \item for each $i \in \{0,1,2\}$ a constructor\\ $\mathsf{pos}_i : \Pos(A_1 +^n_{A_0} A_2,\mathsf{po}(f,g),A_i)$.
  \end{itemize}
  More generally, this works for arbitrary HITs with an additional constructor ensuring $n$-truncatedness.
\end{example}

To establish the $n$-truncatedness of the universe we will need to know
that the extension of the container $\cExt_{S,P}(F,i)$ preserves the truncation level of the family $F$.
We cannot however just ask for $\tlevel{n}(\sm{i :
  I}. \cExt_{S,P}(F,i))$ to hold whenever $\tlevel{n}(\sm{(i : I)}
F(i))$ holds, as the latter would already be the whole result when $F
\equiv \inV$. We extract then the following condition from what is needed during the induction in the proof of \cref{general-V-truncated}.
\begin{definition}[Retaining $n$-truncatedness] \label{trunc-preserve}
  An $I$-indexed container $(S,\Pos)$ \emph{retains} $n$-truncation,
  if for any family $F$ over $I$, any $i_b : I$ and element of the extension $(s_b,t_b) : \cExt_{S,\Pos}(F,i_b)$ for $b \in \{0,1\}$ we have that\\
  $\begin{array}{l}
    \prd{j_0\,j_1 : I}\prd{p_0 : \Pos(i_0,s_0,j_0),p_1 : \Pos(i_1,s_1,j_1)}\\
    \hspace{50pt}\tlevel{(n-1)}(\sm {q : j_0 = j_1} t_{j_0}(p_0) =_{F(q)} t_{j_1}(p_1))\\
    \end{array}$\\
  implies\\
  $\begin{array}{l}\tlevel{(n-1)}((i_0,s_0,t_0) =_{\sm{i : I} \cExt_{S,\Pos}(F,i)} (i_1,s_1,t_1))\end{array}$
\end{definition}
%
\begin{example}[Signatures retaining $n$-truncatedness] \label{ex:trunc-preserve}
  \cref{ex:nullary,ex:pi-cont,ex:po-cont} all retain $n$-truncatedness.
  The case of nullary type formers is trivial. 
  As mentioned the case for $\Pi$-types follows the reasoning in \cref{simple-V-is-set}.
  For truncated pushouts we can observe that $(X,(s,t))$ of type $\sm{X : \UU^{\leq n}} \cExt_{S,\Pos}(F,X)$ is equivalent to the following data:
  \begin{itemize}
  \item $X : \UU^{\leq n}$
  \item for $i \in \{0,1,2\}$, both $A_i : \UU^{\leq n}$ and $w_i : F(A_i)$
  \item $f : A_0 \to A_1$ and $g : A_0 \to A_2$
  \item $q : X = A_1 +^n_{A_0} A_2$
  \end{itemize}
  then $X$ and $q$ form a contractible pair, the types of $f$ and $g$ are $n$-truncated by construction, so we only have to worry about the $(A_i,w_i)$ pairs.
  But since the $w_i$ are obtained from $t$, those components are handled by the premise given to us.
  
  Coproducts of such containers also retain $n$-truncatedness as
  their extension will correspond to the sum of the extensions, which
  means we can collect multiple type formers into a single
  $n$-truncatedness preserving indexed container.
\end{example}

\subsection{A $\Pred$-univalent $n$-truncated Universe of $n$-truncated Types}

Now we have everything in place to provide the final version of our universe
\[ \V = \sm {X : \UU^{\leq n}} \inV_{S,\Pos}^{n,\Pred}(X) \]
with $\El_\V : \V \to \UU^{\leq n}$ given by first projection.
We will often omit the sub- and sup- scripts on $\inV$ as they will be clear from context.

The family $\inV$ is defined as follows.
In analogy with the indexed W-type $W_{S,\Pos}$, which one would use for an ordinary closed universe, we use the \emph{$P$-propositional indexed W-type} $\inV \defeq W_{S,\Pos}^P$ (\longshortref{partially-propositional-indexed-W-type}).
The theory of partially propositional indexed W-types is developed in \longshortref{appendix-on-W-type-stuff}.
For convenience, we give here the explicit definition as an indexed higher inductive type.

\begin{definition} \label{general-in-V}
Given a $\UU^{\leq n}$-indexed container $(S,\Pos)$, the family $\inV$ over $\UU^{\leq n}$ is defined as the higher inductive type generated by the following constructors:
\begin{enumerate}[label=(\roman*)]
\item \label{general-in-V:tcon}
given $c : \cExt_{S,\Pos}(\inV,X)$, a constructor $\tcon(c) : \inV(X)$,
\item \label{general-in-V:univalence}
given $X : \UU^{\leq n}$ with $\Pred(X)$, and $w_0, w_1 : \inV(X)$, a path constructor $\uacon(w_0, w_1) : w_0 =_{\inV(X)} w_1$.
\end{enumerate}
\end{definition}

\cref{simple-V-partially-univalent} generalizes to the new setting.
\begin{lemma-qed} \label{general-V-partially-univalent}
  Let $(S,\Pos)$ be a $\UU^{\leq n}$-indexed container, and $\Pred$ a family of propositions over $\UU^{\leq n}$.
  The type $\V$ with $\El_\V : \V \to \UU^{\leq n}$ is $\Pred$-univalent.
\end{lemma-qed}

We unfold here the definition of codes for equality in $W_{S,\Pos}^P$ of \longshortref{indexed-w-types-equality}.

\begin{problem} \label{general-in-V-Eq}
Given an equality $p \co \UU^{\leq n}(X_0, X_1)$ and $w_i : \inV(X_i)$ for $i \in \braces{0, 1}$, we define: 
\begin{itemize}
\item
a type $\Eq_\inV(p, w_0, w_1)$ of \emph{codes of equalities} between $w_0$ and $w_1$ over $p$, as in \cref{fig:general-Eq-tcon}.
\item
such that $\Eq_\inV(p, w_0, w_1)$ is contractible if $\Pred(X_0)$.
\end{itemize}
\end{problem}

\begin{proof}[Construction]
  The definition proceeds by double induction on $w_0$ and $w_1$, defining the pair of $\Eq_\inV(p, w_0, w_1)$ and its conditional contractibility in one go.
  Given $\Pred(X_0)$ the case when
  both $w_i$ are built with $\tcon$ is contractible because it's a join with an inhabited proposition.
  When either $w_0$ or $w_1$ is built by $\uacon$ we again have by univalence that the type of contractible types is contractible.
\end{proof}

From \longshortref{indexed-w-types-equality}, we have the following result.

\begin{proposition-qed} \label{general-equality-in-V}
  Given $p : \UU^{\leq n}(X_0, X_1)$ and $w_i : \inV(X_i)$ for $i \in \braces{0, 1}$, there is an equivalence
  \[
  w_0 =_{\inV(p)} w_1 \simeq \Eq_\inV(p, w_0, w_1)
  .\qedhere\]
\end{proposition-qed}

As in \cref{sec:simple} we define
\[
\Eq_\V(p,\overline{X}_0,\overline{X}_1) \defeq \sm{p : \UU^{\leq n}(X_0, X_1)} \Eq_\inV(p, w_0, w_1)
.\]
for $\overline{X}_i \equiv (X_i,w_i)$ for $i \in \braces{0, 1}$ and derive its equivalence with equality in $\V$.

\begin{figure*}
\[
  \begin{array}{l l}
    \Eq'_\inV(q, (s_0,t_0),(s_1,t_1)) \defeq \Sigma\,(q_s : s_0 =_{S(q)} s_1).\\
    \hspace{122pt}(e_t : \prd{Y : \UU^{\leq n}} \prd{p_0,p_1} p_0 =_{\Pos(q,q_s,Y)} p_1 \to \Eq_\inV(\refl,{t_0}_{Y}(p_0),{t_1}_Y(p_1)))
\\
  \end{array}
  \]

  \[
  \begin{array}{l l}
    \Eq_\inV(q : X_0 = X_1, w_0, w_1) \defeq
    \left\{ \begin{array}{l l}\Pred(X_0) \join \Eq'_\inV(q,c_0,c_1) & \mbox{if}\,  w_0 \equiv \tcon(c_0), w_1 \equiv \tcon(c_1)\\
        \mbox{by contractibility} & \mbox{if $w_0$ or $w_1$ given by $\uacon(\ldots)$}\\
        \end{array} \right.
  \end{array}
  \]

  \caption{Definition of $\Eq_\inV$.}
  \label{fig:general-Eq-tcon}
\end{figure*}

\begin{corollary-qed} \label{general-equality-V}
For $\overline{X}_0, \overline{X}_1 : V$, we have
\[
\V(\overline{X}_0, \overline{X}_1) \simeq \Eq_\V(\overline{X}_0, \overline{X}_1)
.\qedhere\]
\end{corollary-qed}

\subsubsection{$\V$ is $n$-truncated}

\begin{theorem} \label{general-V-truncated}
  Let $(S,\Pos)$ be an $n$-truncatedness retaining container.
  If $\Pred(X)$ implies $\tlevel{n-1}(X)$
  then $\mathsf{V}$ is $n$-truncated.
  \begin{proof}
    Given $\overline{X}_i \equiv (X_i,w_i) : \mathsf{V}$ for $i \in \{0,1\}$ we proceed by induction on $w_0$ and $w_1$ to prove
    $\V(\overline{X}_0,\overline{X}_1)$ is $(n-1)$-truncated.
    By \cref{general-equality-V} and the same reasoning as in the proof of \cref{simple-V-is-set},
    we have to concern ourselves only with the following pushout square\footnote{a variant of~\cref{simple-V-is-set:2}}:
\begin{equation} \label{general-V-is-truncated:1}
\hspace{-1cm} 
\begin{gathered}
\xymatrix@C-0.5cm{
  \sm{p : \UU^{\leq n}(X_0, X_1)} \Pred(X_0) \times E(p)
  \ar[r]
  \ar@{^(->}[d]
&
  \sm{p : \UU^{\leq n}(X_0, X_1)} \Pred(X_0)
  \ar[d]
\\
  \sm{p : \UU^{\leq n}(X_0, X_1)} E(p)
  \ar[r]
&
  \Eq_V(\overline{X}_0, \overline{X}_1)
  \fancypullback{[u]}{[l(0.5)]}
}
\end{gathered}
\hspace{-1cm} 
\end{equation}
where $E(p) = \Eq'_\inV(p,(s_0,t_0),(s_1,t_1))$.
By \longshortref{pushout-mono-n-truncated}, it is enough to show
the top right and bottom left corners are $(n-1)$-truncated to
conclude that $\Eq_V(\overline{X}_0, \overline{X}_1)$ is as well.
$\sm{p : \UU^{\leq n}(X_0, X_1)} \Pred(X_0)$ is $(n-1)$-truncated
because $\Pred(X_0)$ is a proposition and implies $\tlevel{n-1}(X_0)$,
so that $\UU^{\leq n}(X_0, X_1)$ is $(n-1)$-truncated by univalence.
$\sm{p : \UU^{\leq n}(X_0, X_1)} E(p)$ is equivalent to
$(X_0,(s_0,t_0)) = (X_1,(s_1,t_1))$ by \cref{general-equality-in-V}, so
we can conclude its $(n-1)$-trucatedness by using that $(S,\Pos)$ retains $n$-truncatedness,
because its premise is satisfied by the induction hypothesis.
  \end{proof}
\andrea{``(n-1)-truncatedness'' is probably repeated too much}
\end{theorem}

As the special case where $P$ is constantly false, we obtain the folklore construction of $0$-truncated ``closed'' universes.

\begin{corollary-qed} \label{cor:no-univ-trunc}
  Let $(S,\Pos)$ be a $0$-truncatedness retaining container.
  If $\Pred = \lam{X} \bot$ then $\mathsf{V}$ is $0$-truncated.
\end{corollary-qed}

\section{Models of $n$-truncated Type Theory with Univalence for $(n-1)$-types.}
\label{sec:model}

In this section, we show that Martin-L\"of type theory with function extensionality and the assumption that all types are $n$-truncated is consistent with univalence for $(n-1)$-types.

\andrea{reviewer 2 complains that the phrasing is a bit convoluted, maybe too many details and parentheticals?}
\begin{remark} \label{hierarchies-necessary}
The full strength of this statement is realized only with a sufficiently long chain of universes $\UU_0, \ldots, \UU_k$, one included in the next.
For if $k < n$, it is known~\cite[Section~6]{kraus-sattler:hierarchy-of-universes} how to modify a model of homotopy type theory (including univalence, but no higher inductive types) to be $n$-truncated by restricting types of ``size'' $i$ (classified by $\UU_i$) to $i$-types (and restricting all types to the $n$-truncated).
\end{remark}

For the reason given in the above remark, we consider Martin-L\"of type theory to come with an $\omega$-indexed (cumulative) hierarchy of universes $\UU_0, \UU_1, \ldots$.
Alternatively, we could include higher inductive types, which in the presence of univalence for $(n-1)$-types are still able to produce proper $n$-types.%
\footnote{
For example, univalence for $0$-types is sufficient to show that the circle $S^1$ is a proper $1$-type.
We consider an $n$-type \emph{proper} if it is not an $(n-1)$-type.
}
However, a key point is still for an $n$-truncated universe univalent for $(n-1)$-types to be able to contain a code for a (smaller) universe of the same kind, and at this point we may as well consider a hierarchy of universes.

\smallskip

We use \emph{categories with families} (cwfs)~\cite{dybjer:cwf} as our notion of model of dependent type theory.
They are models of a generalized algebraic theory~\cite{cartmell:generalized-algebraic-theories} (as will all semantic notions considered here).
Fixing the underlying category $\C$, we obtain a category of \emph{cwf structures} on $\C$.
We refer to a cwf structure on $\C$ by its presheaf of types $\Ty$, the presheaf $\Tm$ of terms on its category of elements $\int \Ty$ left implicit.

Let $T$ stand for a choice of type formers, specified by a collection of rules that are generally natural in the context (one way to ensure this naturality is by demanding that these rules be interpretable in presheaves over the category of contexts and substitutions~\cite{awodey:natural-models,paolo:thesis}).
Type formers can be standard type formers such as dependent sums, dependent products, or identity types, but also ``axioms'' such as function extensionality.
As before, we have categories of \emph{cwfs with type formers $T$} as well as \emph{cwf structures with type formers $T$} on a fixed category $\C$.

\begin{definition}[{\cite[Definition~2.4]{annenkov-et-al:two-level-type-theory}}]
A \emph{cwf hierarchy $\Ty$ with type formers $T$} on a category $\C$ is a sequential diagram
\[
\xymatrix{
  \Ty_0
  \ar[r]
&
  \Ty_1
  \ar[r]
&
  \ldots
}
\]
of cwf structures with type formers $T$ on $\C$.
\end{definition}

Note that the \emph{lifting maps} $\Ty_i \to \Ty_{i+1}$ preserve type formers $T$.
As before, cwf hierarchies with type formers $T$ assemble into a category.

\begin{definition}[{\cite[Definition~2.5]{annenkov-et-al:two-level-type-theory}}] \label{model-of-mltt}
A \emph{model of Martin-L\"of type theory} with type formers $T$ is a category $\C$ with a cwf hierarchy $\Ty$ with type formers $T$ on $\C$ together with, for each $i$, a global section $\UU_i$ with an isomorphism $\El_i \co \Tm_{i+1}(\Gamma, \UU_i) \simeq \Ty_i(\Gamma)$ natural in
$\Gamma \in \C$.
\end{definition}

In all uses of the above definition, we will implicitly assume that $T$ contains at least dependent sums, dependent products, identity types, and finite coproducts.
This makes available basic concepts of homotopy type theory such as being $n$-truncated (for an external number $n$).
We write $\MLTT_T$ for the category such models, and $\MLTT_T(\C)$ if we wish to fix the underlying category.
We say that $\C$ is \emph{$n$-truncated} (where $n \geq -2$) if all $A \in \Ty_i(\Gamma)$ are $n$-truncated, naturally in $\Gamma \in \C$, for all $i$.
As a type former (an axiom), we denote it $\Tr(n)$.

We call $\UU_i$ the \emph{$i$-th universe} of $\C$.
Note that, in contrast to our internal reasoning, we explicitly reference the decoding natural transformation $\El_i$.
Given a generic property $P$ of types (such as being $n$-truncated) forming an internal proposition, we say that $\C$ satisfies \emph{univalence for $P$} if the universe $\UU_i$ is $P$-univalent for all $i$ in the sense of what follows \cref{partial-univalence}.
We add a subscript $\UA(n)$ to $\MLTT_T$ to indicate restriction to models with univalence for $n$-types.

Let $T$ be a collection of type formers.
Given $\C \in \MLTT_T$ and $i \geq 0$, the type forming operations contained in $T$ can be encoded as internal operations on the universe $\UU_i$.
Assume that these internal operations restrict to the subuniverse $\UU_i^{\leq n}$ of $n$-types.
For well-behaved $T$, it is possible to find a global $\UU_i$-indexed container $C$ of size $i + 1$ such that for any family $S$ over $\UU_i^{\leq n}$, a lift of the given internal operations on $\UU_i^{\leq n}$ to $\sm{X : \UU_i^{\leq n}} S(X)$ corresponds to a $C$-algebra structure on $S$.
Lastly, assume that $C$ retains $n$-truncatedness in the sense of \cref{trunc-preserve}.
If all of this is the case, naturally in $\C$ , we say that $T$ is \emph{$n$-benign}.

\begin{figure*}
\begin{multicols}{3}
\begin{itemize}
\item
Unit type:
\begin{align*}
S &= 1
,\\
t(\bullet) &= 1
,\\
\Pos(A, x, y) &= 0
.\end{align*}
\item
Dependent sums:
\begin{align*}
S &= \textstyle{\sm{A : \UU^{\leq n}} A \to \UU^{\leq n}}
,\\
t(A, B) &= \textstyle{\sm{a : A} B(a)}
,\\
\Pos(A, B) &= 1 + A
,\\
s((A, B), \inl(\bullet)) &= A
,\\
s((A, B), \inr(a)) &= B(a)
.\end{align*}
\item
Dependent products:
\begin{align*}
S &= \textstyle{\sm{A : \UU^{\leq n}} A \to \UU^{\leq n}}
,\\
t(A, B) &= \textstyle{\prd{a : A} B(a)}
,\\
\Pos(A, B) &= 1 + A
,\\
s((A, B), \inl(\bullet)) &= A
,\\
s((A, B), \inr(a)) &= B(a)
.\end{align*}
\item
Identity types:
\begin{align*}
S &= \textstyle{\sm{A : \UU^{\leq n}} A \times A}
,\\
t(A, x, y) &= A(x, y)
,\\
\Pos(A, x, y) &= 1
,\\
s((A, x, y), \bullet) &= A
.\end{align*}
\item
Empty type:
\begin{align*}
S &= 1
,\\
t(\bullet) &= 0
,\\
\Pos(A, x, y) &= 0
.\end{align*}
\item
Binary coproducts:
\begin{align*}
S &= \UU^{\leq n} \times \UU^{\leq n}
,\\
t(A, B) &= A + B
,\\
\Pos(A, B) &= 1 + 1
,\\
s((A, B), \inl(\bullet)) &= A
,\\
s((A, B), \inr(\bullet)) &= B
.\end{align*}
\end{itemize}
\end{multicols}
\caption{
$\UU^{\leq n}$-indexed containers for basic type formers.
The specifying data is given in a slightly alternate form: a type of shapes $S$, a target function $t : S \to \UU_i^{\leq n}$, a family of positions $\Pos$ over $S$, and a source function $s : \prod_{s : S} \Pos(s) \to \UU_i^{\leq n}$.
This corresponds to a polynomial functor $\UU^{\leq n} \leftarrow S \to P \to \UU^{\leq n}$ with middle arrow a fibration.
The actual indexed container is obtained by taking fibers using the identity type.
}
\label{fig:containers-for-type-formers}
\end{figure*}

\begin{example}
The basic type formers we implicitly require for a model of Martin-L\"of type theory are $n$-benign for any $n \geq 0$.
The associated $\UU_i^{\leq n}$-indexed containers are listed in \cref{fig:containers-for-type-formers} (with the size index $i$ omitted).
Retention of $n$-truncatedness in the sense of \cref{trunc-preserve} follows the scheme of \cref{ex:trunc-preserve}.
\end{example}

\begin{example}
Any type former that only has term forming operations is automatically $n$-benign.
In the first place, this applies to axiom-style type formers such as function extensionality.
\end{example}

We are now ready to state the main result.

\begin{theorem} \label{derived-truncated-model}
Let $n \geq 0$ and $T$ be an $n$-benign choice of type formers, including function extensionality.
Let $(\C, \Ty)$ be a model of Martin-L\"of type theory with type formers $T$ that is univalent for $(n-1)$-types.
Then there is an $n$-truncated model $\Ty'$ of Martin-L\"of type theory with type formers $T$ on $\C$ that is univalent for $(n-1)$-types.
Furthermore, there is a morphism $\Ty' \to \Ty$ of cwf hierarchies with type formers $T$ on $\C$.
\end{theorem}

\begin{proof}
Given a cwf structure $\Ty$ on a category $\C$, note that a further cwf structure $\Ty'$ together with a morphism $\Ty' \to \Ty$ corresponds up to isomorphism to just a presheaf $\Ty'$ of types with a natural transformation $\Ty' \to \Ty$.%
\footnote{This is immediate when switching from cwfs to the equivalent notion of \emph{categories with attributes}.}
The terms of $\Ty'$ are inherited (up to isomorphism) from those of $\Ty$ since terms correspond to sections of context projections and $\Ty' \to \Ty$ should preserve context extension.
Abstractly speaking, the forgetful functor from cwf structures on $\C$ to discrete fibrations on $\C$ is itself a discrete fibration.

Let us further assume that $\Ty$ implements some type type formers $T$.
To interpret $T$ in $\Ty'$ such that $\Ty' \to \Ty$ preserves $T$, we only have to interpret the actual type forming operations of $T$ in $\Ty'$ such that they are preserved by $\Ty' \to \Ty$; the term forming operations of $T$ will then be uniquely inherited from $\Ty$.

Let us now return to the situation of \cref{derived-truncated-model}.
The type forming operations of the type formers $T$ in $(\C, \Ty_i)$ can be encoded as internal operations on the universe $\UU_i$.%
\footnote{Note that this is only a bijective correspondence if we have the judgmental $\eta$-law for dependent products, but this is not required here.}
Since $T$ is $n$-benign, these internal operations further restrict to the subuniverse $\UU_i^{\leq n}$ of $n$-types.
We now wish to define a global type $V_i$ of size $i+1$ with a map $V_i \to \UU_i^{\leq n}$.
Restricting $\El_i$ along this map, we can see $V_i$ as a universe.
Defining $\Ty'_i(\Gamma) = \Tm(\Gamma, V_i)$ with $\Ty'_i \to \Ty_i$ induced by $V_i \to \UU_i^{\leq n}$, we then obtain the cwf structure $\Ty'_i$ with a map $\Ty'_i \to \Ty_i$.
Interpreting $T$ in $\Ty'_i$ compatible with $\Ty_i$ will follow from a (strict) lift of the internal type formation operations from $\UU_i^{\leq n}$ to $V_i$.
Finally, everything needs to be natural in $i \in \omega$.

Let us start with the base $i = 0$.
We will define $V_0 \defeq \sm{X : \UU_0^{\leq n}} \inV_0$ for a family $\inV_0$ over $\UU_0^{\leq n}$, with $V_0 \to \UU_0^{\leq n}$ the first projection.
Using that $T$ is $n$-benign, we have a $\UU_0^{\leq n}$-indexed container $C_0$ such that a lift of the internal formations operations from $\UU_0^{\leq n}$ to $V_0$ corresponds to a $C_0$-algebra structure on $\inV_0$.
We now follow \cref{sec:general} for the construction of $\inV_0$ from $C_0$; this means $\inV_0$ is the $\tlevel{n-1}$-propositional indexed W-type $W_{C_0}^{\tlevel{n-1}}$ as per \longshortref{partially-propositional-indexed-W-types}.
In particular, we obtain a $C_0$-algebra structure on $\inV_0$.
Note that $V_0$ is univalent for $(n-1)$-types by \cref{general-V-partially-univalent} and $n$-truncated by \cref{general-V-truncated}.

For general $i$, we let $C_i'$ be the coproduct of the $\UU_i^{\leq n}$-indexed container $C_i$ given from $T$ being $n$-benign with the indexed container with shapes $0 \leq j < i$, with indexing of $j$ being $V_j$ (lifted to $\Ty_i$), and no positions.
We then define $\inV_i$ from $C_i'$ as before.
This guarantees that there are codes for the universes below $i$ in $V_i$.

To make $\Ty'$ into a cwf hierarchy and $\Ty' \to \Ty$ into a morphism of cwf hierachies, we need to construct, for every $i \geq 0$, a dotted morphism making the naturality square
\[
\xymatrix{
  \Ty'_i
  \ar[r]
  \ar@{.>}[d]
&
  \Ty_i
  \ar[d]
\\
  \Ty'_{i+1}
  \ar[r]
&
  \Ty_{i+1}
}
\]
of presheaves of types commute.
This amounts to defining internal $V_i \to V_{i+1}$ making the square
\[
\xymatrix{
  V_i
  \ar[r]
  \ar@{.>}[d]
&
  \UU_i
  \ar[d]^{\lift}
\\
  V_{i+1}
  \ar[r]
&
  \UU_{i+1}
}
\]
commute strictly.
In turns, this corresponds to an internal function
\begin{equation} \label{main-theorem:internal-lift'}
\begin{gathered}
\inV_i(X) \to \inV_{i+1}(\mathsf{lift}(X))
\end{gathered}
\end{equation}
for $X : \UU_i^{\leq n}$.
We define this by recursion for $\inV_i$, noting that $\inV_{i+1}$ restricted along $\mathsf{lift}$ carries a $C_i'$-algebra structure, forgetting the code for the $i$-th universe in its $C_{i+1}'$-algebra structure.

It remains to check that the map $\Ty'_i \to \Ty'_{o+1}$ respects the type forming operations of $T$.
This follows from~\cref{main-theorem:internal-lift'} commuting strictly with $C_i$-algebra structures.
This follows from the judgmental $\beta$-law of the higher inductive family $\inV_i$.%
\footnote{
This is the only place in our entire construction where judgmental $\beta$-laws for higher inductive types are needed.
One might well regard it as an artifact of our universe hierarchy setup.
}

It remains to check that $\Ty'$ has universes as required by \cref{model-of-mltt}.
Indeed, the $i$-th universe $\UU_i'$ is simply given by $V_i$ itself, with $\El'_i$ the identity isomorphism.
\end{proof}

\begin{corollary} \label{relative-consistency}
Relative to Martin-L\"of type theory with function extensionality and univalence for $(n-1)$-types (and any further $n$-benign type formers), if the addition of pushouts and propositionally truncated indexed W-types is consistent, then it is consistent to assume that all types are $n$-truncated.
\end{corollary}

\begin{proof}
Given a model for the former theory, we obtain a model (on the same category) of the latter theory by \cref{derived-truncated-model}.
By construction, the empty type is inhabited in this model exactly if it is inhabited in the old model.
\end{proof}

\begin{corollary} \label{consistency-result}
In Martin-L\"of type theory with function extensionality and univalence for $(n-1)$-types (and any further $n$-benign type formers implemented by a known model of homotopy type theory), it is consistent to assume that all types $n$-truncated.
\end{corollary}

\begin{proof}
Apply \cref{relative-consistency} to a model of homotopy type theory such as simplicial sets or cubical sets that supports higher inductive families.
\end{proof}

\section{A Cubical Type Theory with UIP, Propositional Extensionality, and Homotopy Canonicity}
\label{sec:htcan}

In \cite{coquand-huber-sattler:homotopy-canonicity} the authors
establish the homotopy canonicity property for a cubical type theory
without judgmental equations for the box filling operations.
Here we will follow that proof to prove homotopy canonicity for a
cubical type theory with an axiomatic UIP principle and propositional
extensionality given by a modified $\Glue$-type.
To keep this section brief, we closely follow their notation.

\subsection{$0$-truncated Cubical Cwf}
We take the definition of cubical cwf from
\cite{coquand-huber-sattler:homotopy-canonicity} and adapt it by
adding a new $\mathsf{trunc}$ operation and an extra argument to $\Glue$-types.
The definition is internal to the category of cubical sets of \cite{CCHM}.
A minor difference to \cite{coquand-huber-sattler:homotopy-canonicity}, following the previous section, we only require a sequential diagram of $\Ty_i$ presheaves, without topmost $\Ty$, and we do not require the lifting map $\Ty_i \to \Ty_{i+1}$ to be mono.

Given $A : \Ty_i(\Gamma)$, we define $\isprop(A)$ and $\isset(A)$ in $\Ty_i(\Gamma)$ using the $\mathsf{Path}$-type former.

\begin{itemize}
\item \textbf{Glue types.} Given $A : \Ty_i(\Gamma)$, $A_p : \Elem(\Gamma, \isprop(A))$, $\varphi : \FF$,
  $T : [\varphi] \to \Ty_i(\Gamma)$, and $e : \Elem(\Gamma,\Equiv(T\,\tt,A))$, we have the \emph{glueing} $\Glue(A,A_p,\varphi,T,e)$ in $\Ty_i(\Gamma)$, equal to $T\,\tt$ on $\varphi$.
  We also have $\mathsf{glue}(a,t)$, $\mathsf{unglue}$, and their equations as described in \cite[Sec.~1.3]{coquand-huber-sattler:homotopy-canonicity}.
\item \textbf{$0$-truncation operation.}
  Given $A$ in $\Ty_i(\Gamma)$ we have $\mathsf{trunc}(A)$ in $\Elem(\Gamma,\isset(A))$.
  No equations are required other than stability under substitution.
\end{itemize}

\subsection{Standard Model}
We now work in the category of cubical sets of \cite{CCHM}.
It satisfies the assumptions listed at the top of
\cite[Section~2.2]{coquand-huber-sattler:homotopy-canonicity}, so we have
a hierarchy of universes of fibrant types $\UFib_i$, defined from a
cumulative hierarchy of universes of presheaves $\mathsf{U}_i$ for $i
\in \{0,1,\ldots,\omega\}$.
Using it as a model of cubical type theory with uniformly indexed higher inductive types, we replay the construction from \cref{sec:general} and obtain a family $\inV : (\UFib_i)^{\leq 0} \to \UFib_{i+1}$.
Then we take
\[
\V_i \equiv \textstyle{\sm{A : (\UFib_i)^{\leq 0}} \inV(A) : \UFib_{i+1}}
\]
with $\El_{\V_i} : \V \to (\UFib_i)^{\leq 0}$.

Just as~\cite[Section~2.3]{coquand-huber-sattler:homotopy-canonicity} defines the standard model as an internal cwf from the universes $\UFib_i$, we define the standard model from the universes $V_i$:
\begin{itemize}
\item $\mathsf{Con}$ is the category with objects in $\mathsf{U}_\omega$ and functions between them as morphism,
\item the types of size $i$ over $\Gamma : \mathsf{U}_\omega$ are maps $\Gamma \to \V_i$,
\item the elements of $A : \Gamma \to \V_i$ are $\Pi(\rho : \Gamma). \El_{\V_i}(A\,\rho)$.
\end{itemize}
We will often omit the use of $\El_{\V_i}$ to lighten the notational burden.

Type formers $\Pi,\Sigma,\mathsf{N},\mathsf{Path}$ and universes are
given by including a code for them in the container $(S,\Pos)$ for
$\V_i$. The filling operation is derived from the one for $\UFib_i$.
$\Glue$-types are handled below.

\subsubsection{A Code for $\Glue$ in $\V_i$}

One would think that $\inV$ might need an explicit constructor for $\Glue$.
However, the path constructor $\uacon$ of $\inV$ suffices to derive one, given
$\Glue$ for $\UFib$ and fibrancy of $\inV$.

Given $\Gamma : \mathsf{U}_\omega$, $A : \Gamma \to \V_i$, $A_p :
\Pi(\rho : \Gamma). \isprop(A\,\rho)$, $\varphi : \FF$, $T : [\varphi]
\to \Gamma \to \V_i$ and $e : [\varphi] \to \mathsf{Equiv}(T \tt,
A)$, we wish to define $\Glue(A,A_p,\varphi,T,e) : \Gamma \to \V_i$.
We take
\[
\Glue(A,A_p,\varphi,T,e)\,\rho \defeq (G,w_G)
\]
where $G = \Glue(\El_\V(A\,\rho),\varphi,\El_\V \circ (T\,\rho),e\,\rho) : \UFib_i$, and given $w_A \equiv A\,\rho .2$ and $w_T \equiv \lam{o} T\,o\,\rho .2$,
we obtain $w_\Glue$ by first transporting $w_A :
\inV(\El_\V(A\,\rho))$ to $w_G' : \inV(G)$ by the canonical path between the two indices, and then
composing under $[\varphi]$ with a path between $w_G'$ and $w_T$ built
by $\uacon$. The latter is possible because, assuming $[\varphi]$, both
$w_G'$ and $w_T$ are codes for $\El_V(T\,\tt\,\rho)$, which is propositional
by $A_p\,\rho$.
Note that with this correction $\Glue(A,A_p,\varphi,T,e) \equiv T\,\tt$ when $[\varphi] \equiv \top$.
One then checks that the code so defined commutes with the lifting maps $\V_i \to \V_{i+1}$.

\subsection{Sconing Model}

\newcommand{\MM}{\mathcal{M}}

Given a cubical cwf $\MM$ (denoted by $\mathsf{Con}$, $\Ty_i$, $\Elem$, $\ldots$), we want to define a new cubical cwf
$\mathcal{M}^*$, (denoted by $\mathsf{Con}^*$, $\Ty^*_i$, $\Elem^*$, $\ldots$) as the Arting glueing of $\mathcal{M}$ along an internal global
sections functor $\verts{-}$. We assume $\MM$ size-compatible with the universes $\mathsf{U}_i$ as in \cite[Sec.~3]{coquand-huber-sattler:homotopy-canonicity}.
In \cite{coquand-huber-sattler:homotopy-canonicity}, the functor $\verts{-}$ targets the standard model directly, given that $\Elem(1,A)$ is a fibrant type.
In our case, we have to include a code for it in $\inV_i$, as in extending the container $(S,\Pos)$, as follows:
\begin{itemize}
  \item given $A : \Ty_i(1)$, a constructor $\ceils{A} : \inV_i(\Elem(1,A))$.
\end{itemize}
Note that both $\Ty_i(1)$ and $\Elem(1,A)$ are $0$-truncated, because
the $\mathsf{trunc}$ operation implies $\isset(\Elem(1,A))$ for any
$A$ in $\Ty_j(1)$, and $\Ty_i(1)$ itself is equivalent to
$\Elem(1,\mathsf{U}_i)$.  This makes sure that the extended container
still preserves $0$-truncatedness.

For the definition of natural numbers in $\MM^*$, we will also need a
code for the type family $\mathbb{N}' : \Elem(1,\mathsf{N}) \to
\UFib_0$ defined in \cite[Appendix B]{coquand-huber-sattler:homotopy-canonicity}:
\begin{itemize}
  \item given $n : \Elem(1,\mathsf{N})$, a constructor $\mathsf{N}'(n) : \inV_0(\mathbb{N}'\,n)$
\end{itemize}
where $\mathbb{N}'\,n$ can be shown to be $0$-truncated by \cref{cor:no-univ-trunc}.

The functor $\verts{-}$ is then given on contexts, types, and elements of $\mathcal{M}$ like so:
\begin{itemize}
\item $|\Gamma| \defeq \mathsf{Hom}_\mathcal{M}(1,\Gamma)$,
\item $|A|\,\rho \defeq (\Elem(1, A), \ceils{A\rho})$,
\item $|a|\,\rho \defeq a\rho$.
\end{itemize}

We now define the sconing model $\MM^*$, starting with the cwf components.
\begin{itemize}
\item A context $(\Gamma,\Gamma') : \Con^*$ consists of $\Gamma : \Con$ in $\MM$ and a family $\Gamma' : |\Gamma| \to \UU_\omega$.
\item A type $(A,A') : \Ty_i^*(\Gamma,\Gamma')$ consists of a type $A : \Ty_i(\Gamma)$ in $\MM$ and a family
  \[ A' : \Pi (\rho : |\Gamma|) (\rho' : \Gamma'\,\rho) \to |A|\,\rho \to \V_i \]
  of proof-relevant predicates over it.
\item An element $(a,a') : \Elem^*((\Gamma,\Gamma'),(A,A'))$ consists of an element $a : \Elem(\Gamma,A)$ in $\MM$ and
  \[ a' : \Pi (\rho : |\Gamma|) (\rho' : \Gamma'\,\rho) \to A(\rho,\rho',a\rho) .\]
\end{itemize}
We observe that this definitions differs from the one given in
\citet{coquand-huber-sattler:homotopy-canonicity} only by the use of
$\V_i$ in place of $\UFib_i$ to define $\Ty_i^*$. As such we will
not repeat here the details about the rest of the cwf structure or the
shared type formers and operations, and instead discuss only $\Glue^*$ and $\trunc^*$.

\subsubsection{$\Glue$-types}

\begin{lemma} \label{lem:isprop-sconing}
  Let $(A,A')$ in $\Ty_i^*(\Gamma,\Gamma')$. The following statements are logically equivalent, naturally in $(\Gamma,\Gamma')$:
\begin{gather}
\Elem^*((\Gamma,\Gamma'), \isprop^*(A,A'))
\label{lem:isprop-sconing:0}
\\
\begin{split}
&\Elem(\Gamma, \isprop(A))
\\[-0.6ex]
{} \times {} &\Pi(\rho : |\Gamma|)(\rho' : \Gamma'\,\rho). \isprop(\Sigma (a : |A|\,\rho). A'(\rho,\rho'\,a))
\end{split}
\label{lem:isprop-sconing:1}
\\
\begin{split}
&\Elem(\Gamma, \isprop(A))
\\[-0.6ex]
{} \times {} &\Pi(\rho : |\Gamma|)(\rho' : \Gamma'\,\rho) (a : |A|\,\rho). \isprop(A'(\rho,\rho'\,a))
\label{lem:isprop-sconing:2}
\end{split}
\end{gather}
\end{lemma}

\begin{proof}
Given~\cref{lem:isprop-sconing:0}, we have $A_p : \Elem(\Gamma,\isprop(A))$ and a proof that one can fill lines in $A'$ over lines produced by $|A_p|$; that is enough to fill lines in the $\Sigma$-type in~\cref{lem:isprop-sconing:1}.
From there, we derive~\cref{lem:isprop-sconing:2}: since $|A|\,\rho$ is propositional, any path from $a$ to $a$ is constant.
Going back to~\cref{lem:isprop-sconing:0} requires only to contract a path in $|A|\,\rho$.
\end{proof}

Let $(A,A')$ in $\Ty_i^*(\Gamma,\Gamma')$, $(A_p,A_p')$ in $\Elem^*((\Gamma,\Gamma'), \isprop^*(A,A'))$, $\varphi$ in $\FF$,
$\angles{T,T'}$ in $[\varphi] \to \Ty_i^*(\Gamma,\Gamma')$, and $\angles{e,e'}$ in\\ $\Elem^*((\Gamma,\Gamma'),\Equiv^*((T\,\tt,T'\,\tt),(A,A')))$.
We follow the recipe of \cite[Sec.~3.2.6]{coquand-huber-sattler:homotopy-canonicity} and define
\[ \Glue^*((A,A'),(A_p,A_p'),\varphi,\angles{T,T'},\angles{e,e'}) \] as $(\Glue(A,A_p,\varphi,T,e),G')$ where $G'\,\rho\,\rho'\,(\glue(a,t))$ is defined as the $\Glue$-type in $\V_i$ between $A'\,\rho\,\rho'\,a$ and $T'\,\tt\,\rho\,\rho'\,(t\,\tt)$ along $\varphi$. In our case we also have to provide a proof of $\isprop(A'\,\rho\,\rho'\,a)$, which we obtain from $(A_p,A_p')$ by applying \cref{lem:isprop-sconing}, going from~\cref{lem:isprop-sconing:0} to~\cref{lem:isprop-sconing:2}.

\subsubsection{$\mathsf{trunc}$-operation}

\begin{lemma} \label{lem:isset-sconing}
Let $(A,A') : \Ty_i^*(\Gamma,\Gamma')$. The following statements are logically equivalent, naturally in $(\Gamma,\Gamma')$:
\begin{gather}
\Elem^*((\Gamma,\Gamma'), \isset^*(A,A'))
\label{lem:isset-sconing:0}
\\
\begin{split}
& \Elem(\Gamma, \isset(A))
\\[-0.6ex]
{} \times {}& \Pi(\rho : |\Gamma|)(\rho' : \Gamma'\,\rho). \isset(\Sigma (a : |A|\,\rho). A'(\rho,\rho'\,a))
\end{split}
\label{lem:isset-sconing:1}
\\
\begin{split}
& \Elem(\Gamma, \isset(A))
\\[-0.6ex]
{} \times {}& \Pi(\rho : |\Gamma|)(\rho' : \Gamma'\,\rho) (a : |A|\,\rho). \isset(A'(\rho,\rho'\,a))
\end{split}
\label{lem:isset-sconing:2}
\end{gather}
\end{lemma}

\begin{proof}
This follows the same strategy as \cref{lem:isprop-sconing}, except this time filling and contracting squares rather than lines.
\end{proof}

Let $(A,A') : \Ty_i^*(\Gamma,\Gamma')$.
We define
\[
\mathsf{trunc}^*(A,A') : \Elem^*((\Gamma,\Gamma'),\isset^*(A,A))
\]
by applying \cref{lem:isset-sconing} in the direction from~\cref{lem:isset-sconing:0} to~\cref{lem:isset-sconing:2}. to the pair of $\mathsf{trunc}(A)$ and
$\mathit{trunc}'$ where \[\mathit{trunc'}\,\rho\,\rho'\,a : \isset(\El_{\V_i}(A'(\rho,\rho'\,a)))\] is given by $\El_{\V_i}(A'(\rho,\rho'\,a)) : (\UFib_i)^{\leq 0}$.

Given the above constructions, one mechanically verifies the necessary laws to obtain the following statement.
\begin{theorem-qed}[Sconing]
  Given any $0$-truncated cubical cwf $\MM$ that is size-compatible in the sense of \cite[Sec.~3]{coquand-huber-sattler:homotopy-canonicity},
the sconing $\MM^*$ is a $0$-truncated cubical cwf with operations defined as above. We further have a morphism $\MM^* \to \MM$ of $0$-truncated cubical cwfs given by the first projection.
\end{theorem-qed}

We state homotopy canonicity with reference to the initial $0$-truncated cubical cwf $\mathcal{I}$, whose existence and size-compatibility is justified as in \cite[Sec.~4]{coquand-huber-sattler:homotopy-canonicity}. The proof of the theorem also follows the argument given in that section.

\begin{theorem-qed}[Homotopy canonicity]
 In the internal language of the cubical sets category of \cite{CCHM}, given a closed natural $n : \Elem(1,\mathsf{N})$ in the initial model $\mathcal{I}$, we have a numeral $k : \mathbb{N}$ with $p : \Elem(1,\mathsf{Path}(\mathsf{N},n,\mathsf{S}^k(0)))$.
\end{theorem-qed}

\andrea{TODO: something about extending to strict canonicity}

\section{Related Work and Conclusion}
\label{sec:conclusion}

\subsection{Related Work}
In the realm of type theories with UIP and function extensionality,
XTT \cite{sterling-angiuli-gratzer:XTT} is a non-univalent variant
of CTT that takes the extra step of making UIP hold judgmentally, in
the spirit of observational type theory (OTT) \cite{OTT}. As
formulated XTT does not provide propositional extensionality and
requires a typecase operation within the theory for (strict)
canonicity. OTT does include propositional extensionality, but only for a
universe of propositions closed under a specific set of type formers
that made it possible to assume judgmental proof irrelevance for
such propositions.  We conjecture that by introducing UIP (or
$n$-truncatedness) only as a path equality we will be able to refine
our theory to one with strict canonicity without encountering
similar limitations.
Regarding strict propositions, i.e. where any two
elements are stricly equal in the model, the semantics for a
univalent universe of them within the cubical sets model is
described in \cite{coquand:bishop}.
However the corresponding universe of strict sets is not a strict set
itself.
Such semantics are used in \cite{coq:sprop} to justify the
addition of a primitive universe of strict propositions
$\mathsf{sProp}$.

\subsection{Conclusion}
We proved consistency for a theory with $n$-truncatedness and univalence for $(n-1)$-types.
We also showed homotopy canonicity for cubical variant of such a theory.
The main technical tool used was an $n$-truncated universe of $n$-types that is also univalent for $(n-1)$-types.
We would like to stress that such a universe can also be used directly
in HoTT with indexed higher inductive types, for applications that do
not mind the universe being limited to a fixed set of type formers.


\begin{acks}
  Christian Sattler was supported by USAF grant FA9550-16-1-0029.
  Andrea Vezzosi was supported by a research grant (13156) from VILLUM FONDEN.
\end{acks}

\bibliography{lics20}

\INLONGVERSION{
\appendix
\section{Pushouts Along Monomorphisms}
\label{pushouts-along-monomorphisms}

In this subsection, we prove some useful statements about pushouts along monomorphisms, in particular the fact that $n$-truncatedness of objects is preserved (\cref{pushout-mono-n-truncated}).
We were unable to locate this statement in the higher topos literature.
We expect that in homotopy type theory, many of these statements can also be obtained  from~\cite{kraus-raumer:quotient-equality} (pushouts being an example of a non-recursive higher inductive type), albeit under less minimalistic assumptions.

We use the language of higher categories (modelled for example by quasicategories~\cite{joyal-quaderns,lurie:htt}).
The statements of importance for the main body are about locally Cartesian closed higher categories.
In homotopy type theory, this corresponds to the presence of identity types, dependent sums, and dependent products with function extensionality.
All these statements and their proofs can also be read in homotopy type theory with these type formers, seen as an internal language for locally Cartesian closed higher categories (with pushout squares treated axiomatically as in~\cite{egbert:hott-intro}).
It is in this form that they are used in the main body of the paper.

A map $V \to U$ in a higher category $\C$ is \emph{univalent} if for any map $Y \to X$, the space of pullback squares from $Y \to X$ to $V \to U$ is $(-1)$-truncated, \ie if the object $V \to U$ is $(-1)$-truncated in the higher category $\C^\to_\cart$ of arrows and Cartesian morphisms.%
\footnote{We choose this definition as it makes sense in any higher category $\C$.}
We note the following for locally Cartesian closed $\C$.
\begin{itemize}
\item
Given univalent $V \to U$ and an object $Z$, the pullback $Z \times V \to Z \times U$ is univalent in the slice over $Z$.
This is a consequence of pullback pasting.
\item
The internal object of pullback squares (constructed using exponentials) from $Y \to X$ to univalent $V \to U$ is $(-1)$-truncated.
To see this, one considers the space of pullback squares from $Z \times Y \to Z \times X$ to $V \to U$ for an arbitrary object $Z$.
\end{itemize}
Combining these observations, we may phrase univalence via the generic case: in the slice over $U$, the internal object of pullback squares from $V \to U$ to $U \times V \to U \times U$ (living over $U$ via the first projection) is $(-1)$-truncated, or equivalently terminal.%
\footnote{
The reverse implication proceeds as follows.
Given $Y \to X$, let us show that the space of pullback squares from $Y \to X$ to $V \to U$ is $(-1)$-truncated.
We may suppose that $Y \to X$ is classified by $V \to U$ via a map $X \to U$.
The desired space is isomorphic to the space of pullback squares over $U$ from $Y \to X$ to $U \times V \to U \times U$.
Over $U$, the map $Y \to X$ is the product of $V \to U$ with $X$, so this space is the hom space over $U$ from $X$ to the internal object of pullback squares from $V \to U$ to $U \times V \to U \times U$.
The target object of this hom space was assumed terminal.
}
This is the univalence axiom for the ``universe'' $U$ and the ``universal family of elements'' $V \to U$ in homotopy type theory: given $x : U$, the dependent sum of $y : U$ and $V(x) \simeq V(y)$ is contractible, in turn equivalent to \cref{univalence}.
It also corresponds to the definition of univalent family given in~\cite{gepner-kock:univalence}.
The connection between univalent universes in the original sense of Voevodsky and univalent maps in our sense was probably first recognized by Joyal.

A \emph{classifier} for a collection of maps $Y_i \to X_i$ for $i \in I$ is a univalent map $V \to U$ having $Y_i \to X_i$ as pullback for $i \in I$:
\[
\xymatrix{
  Y_i
  \ar[r]
  \ar[d]
  \fancypullback{[d]}{[r]}
&
  V
  \ar[d]
\\
  X_i
  \ar[r]
&
  U
\rlap{.}}
\]
An object $X$ is said to be classified if the map $X \to 1$ is classified.
In homotopy type theory, a classifier for $(f_i)_{i \in I}$ (with external indexing) is simply a family $V$ over $U$ that satisfies the univalence axiom and restricts (up to equivalence) to the given families $f_i$ for $i \in I$.

As discussed in~\cite{lurie:htt}, classifiers are related to \emph{descent} for colimits in the sense of~\cite{rezk:model-topos}.
Here, we are interested only in the case of pushouts.
Instead of assuming blanket classifiers for classes of ``small'' maps as in a higher topos, we will keep track of which collections of maps need to have classifiers.

\begin{proposition}[Pushout descent from classifier] \label{descent}
Let $\C$ be a locally Cartesian closed higher category.
Let
\begin{equation} \label{descent:0}
\begin{gathered}
\xymatrix@!C@C-0.2cm{
  Y_{00}
  \ar[rr]
  \ar[dd]_{p_{00}}
  \ar[dr]
  \ar[rr]
  \fancypullback{[dd]}{[rrr]}
  \fancypullback{[dd]}{[dr]}
&&
  Y_{01}
  \ar[dd]|!{[dl];[dr]}{\hole}_(0.3){p_{01}}
  \ar[dr]
&\\&
  Y_{10}
  \ar[rr]
  \ar[dd]_(0.3){p_{10}}
&&
  Y_{11}
  \ar[dd]_{p_{11}}
  \fancypullback{[ll]}{[ul]}
\\
  X_{00}
  \ar[rr]|!{[ur];[dr]}{\hole}
  \ar[dr]
&&
  X_{01}
  \ar[dr]
\\&
  X_{10}
  \ar[rr]
&&
  X_{11}
  \fancypullback{[ll]}{[ul]}
}
\end{gathered}
\end{equation}
be a cube in $\C$ with horizontal faces pushouts and left and back faces pullbacks (as indicated).
Assume that the maps $p_{01}$ and $p_{10}$ have a common classifier $V \to U$.
Then the right and front faces are pullbacks and the map $p_{11}$ is classified by $V \to U$.
\end{proposition}

\begin{proof}
This is a standard exercise, we give the proof for completeness.

Let $u \co V \to U$ be the given classifier.
In $\C^\to_\cart$, by $(-1)$-truncatedness of $u$, we obtain a square
\[
\xymatrix{
  p_{00}
  \ar[r]
  \ar[d]
&
  p_{01}
  \ar@/^1em/[ddr]
\\
  p_{10}
  \ar@/_1em/[drr]
\\&&
  u
\rlap{.}}
\]
On codomains, we obtain an induced map
\[
\xymatrix{
  X_{00}
  \ar[r]
  \ar[d]
&
  X_{01}
  \ar[d]
  \ar@/^1em/[ddr]
\\
  X_{10}
  \ar[r]
  \ar@/_1em/[drr]
&
  X_{11}
  \ar@{.>}[dr]
  \fancypullback{[u]}{[l]}
\\&&
  U
\rlap{.}}
\]
This allows us to see the bottom face of~\cref{descent:0} as a square in the slice over $U$.
Pulling back along $u$, we obtain a cube
\begin{equation} \label{descent:1}
\begin{gathered}
\xymatrix@!C@C-0.2cm{
  Y_{00}
  \ar[rr]
  \ar[dd]_{p_{00}}
  \ar[dr]
  \ar[rr]
  \fancypullback{[dd]}{[rrr]}
  \fancypullback{[dd]}{[dr]}
&&
  Y_{01}
  \ar[dd]|!{[dl];[dr]}{\hole}_(0.3){p_{01}}
  \ar[dr]
  \fancypullback{[dd]}{[dr]}
&\\&
  Y_{10}
  \ar[rr]
  \ar[dd]_(0.3){p_{10}}
  \fancypullback{[dd]}{[rr]}
&&
  Y_{11}'
  \ar[dd]
  \fancypullback{[ll]}{[ul]}
\\
  X_{00}
  \ar[rr]|!{[ur];[dr]}{\hole}
  \ar[dr]
&&
  X_{01}
  \ar[dr]
\\&
  X_{10}
  \ar[rr]
&&
  X_{11}
  \fancypullback{[ll]}{[ul]}
\rlap{.}}
\end{gathered}
\end{equation}
Here, the top square is a pushout since pullback along $u$ preserves pushouts, a consequence of local Cartesian closure.
The cubes~\cref{descent:0,descent:1}, seen from the top as squares in $\C^\to$, are pushouts of the same span $p_{01} \leftarrow p_{00} \rightarrow p_{10}$.
Pushouts are unique up to isomorphism, so the cubes are isomorphic.
As the right and front faces are pullbacks in~\cref{descent:1}, so are they in~\cref{descent:0}.
By construction, $Y_{11}' \to X_{11}$ is classified by $V \to U$, hence so is $p_{11}$.
\end{proof}

Let us clarify the connection of \cref{descent} to descent.
A finitely complete higher category $\C$ has \emph{descent} for a pushout square
\begin{equation} \label{true-descent}
\begin{gathered}
\xymatrix{
  X_{00}
  \ar[r]
  \ar[d]
&
  X_{01}
  \ar[d]
\\
  X_{10}
  \ar[r]
&
  X_{11}
  \fancypullback{[l]}{[u]}
}
\end{gathered}
\end{equation}
if the slice functor $\C/-$ from $\C^\op$ to higher categories sends it to a pullback square
\[
\xymatrix{
  \C/X_{11}
  \ar[r]
  \ar[d]
  \fancypullback{[d]}{[r]}
&
  \C/X_{10}
  \ar[d]
\\
  \C/X_{01}
  \ar[r]
&
  \C/X_{00}
}
\]
(morphisms are pullback functors).
If this is the case, then:
\begin{itemize}
\item
in any a cube~\cref{descent:0} extending~\cref{true-descent}, the front and back faces are pullbacks (as in \cref{descent}),
\item
if $\C$ has pushouts of pullbacks of $X_{00} \to X_{10}$, the pushout \cref{true-descent} is stable under pullback.
\end{itemize}
If $\C$ has pushouts of pullbacks of $X_{00} \to X_{10}$, these conditions are equivalent to descent for~\cref{true-descent}.

Note that local Cartesian closure implies stability of all colimits, in particular pushouts, under pullback.
In that case, \cref{descent} shows that descent for a pushout follows from the existence of classifiers for finite collections of maps and the existence of certain pushouts.
However, we eschew such a strong assumption on classifiers.

We will directly assume descent for pushouts for a track of results (\cref{pushout-mono-descent}, \cref{pushout-mono-preserves-pullbacks}, \cref{pushout-mono-almost-lex}, \cref{join-prop-level-descent}, \cref{pushout-mono-n-truncated-descent}, and \cref{pushout-mono-n-truncated-descent-restricted}) which cannot easily be expressed in terms of classifiers for some fixed maps or serve as inspiration for the setting using classifiers for a restricted choice of maps.

We will now prove a series of facts (some of them standard) about pushouts along monomorphisms, \ie $(-1)$-truncated maps.
These are sometimes also called embeddings.

\begin{lemma} \label{pushout-mono-descent}
In a finitely complete higher category with descent for pushouts, consider a pushout
\begin{equation} \label{pushout-mono-descent:square}
\begin{gathered}
\xymatrix{
  A
  \ar[r]
  \ar@{^(->}[d]
&
  C
  \ar[d]
\\
  B
  \ar[r]
&
  D
  \fancypullback{[l]}{[u]}
\rlap{.}}
\end{gathered}
\end{equation}
If $A \to B$ mono, then so is $C \to B$ and the square is a pullback.
\end{lemma}

This is standard.
We give its proof so that we can translate it to the setting with limited assumptions on classifiers.

\begin{proof}[Proof of \cref{pushout-mono-descent}]
Consider the cube
\begin{equation} \label{pushout-mono-descent:0}
\begin{gathered}
\xymatrix@!C@C-0.1cm{
  A
  \ar[rr]
  \ar[dd]
  \ar[dr]
  \fancypullback{[dd]}{[rrr]}
  \fancypullback{[dd]}{[dr]}
&&
  C
  \ar[dd]|!{[dl];[dr]}{\hole}
  \ar[dr]
&\\&
  A
  \ar[rr]
  \ar[dd]
&&
  C
  \ar[dd]
  \fancypullback{[ll]}{[ul]}
\\
  A
  \ar[rr]|!{[ur];[dr]}{\hole}
  \ar[dr]
&&
  C
  \ar[dr]
\\&
  B
  \ar[rr]
&&
  D
  \fancypullback{[ll]}{[ul]}
\rlap{.}}
\end{gathered}
\end{equation}
It is given by the functorial action on the map of arrows from $A \to B$ to $C \to D$ of the ``connection operation'' turning an arrow $X \to Y$ into a square
\[
\xymatrix{
  X
  \ar[r]
  \ar[d]
&
  X
  \ar[d]
\\
  X
  \ar[r]
&
  Y
\rlap{.}}
\]
The map $X \to Y$ is mono exactly if this square is a pullback.

Let us inspect the faces of~\cref{pushout-mono-descent:0}.
The bottom face is the original pushout.
The top face is a pushout and the back face is a pullback since opposite edges are invertible.
The left face is a pullback since $A \to B$ is mono.
By descent, it follows that the right and front faces are pullbacks.
From the right face, we infer that $C \to D$ is mono.
From the front face, we infer that~\cref{pushout-mono-descent:square} is a pullback.
\end{proof}

Let us reprove this result starting from a classifier.
Note the extended conclusion.

\begin{lemma} \label{pushout-mono}
In a locally Cartesian closed higher category, let $A \to B$ be a monomorphism admitting a classifier $V \to U$ that also classifies the terminal object.
Then for any pushout
\begin{equation} \label{pushout-mono:square}
\begin{gathered}
\xymatrix{
  A
  \ar[r]
  \ar@{^(->}[d]
&
  C
  \ar[d]
\\
  B
  \ar[r]
&
  D
  \fancypullback{[l]}{[u]}
\rlap{,}}
\end{gathered}
\end{equation}
the map $C \to D$ is mono and classified by $V \to U$ and the square is a pullback.
\end{lemma}

\begin{proof}
We follow the proof of \cref{pushout-mono-descent}.
When it comes to using descent in the cube~\cref{pushout-mono-descent:0}, we apply \cref{descent}.
This uses that $A \to B$ and $C \to C$ (a pullback of $1 \to 1$) are classified by $V \to U$ and additionally gives that $C \to D$ is classified by $V \to U$.
\end{proof}

\cref{pushout-mono-descent} has the following interesting consequence.

\begin{proposition} \label{pushout-mono-preserves-pullbacks}
In a finitely complete higher category with descent for pushouts, pushouts along monomorphisms preserve pullbacks (assuming the involved pushouts exist).
\end{proposition}

\begin{proof}
Recall that the forgetful functor from a coslice creates pullbacks.
Consider a monomorphism $A \to B$ and a pullback square
\[
\xymatrix{
  C_{00}
  \ar[r]
  \ar[d]
  \fancypullback{[d]}{[r]}
&
  C_{01}
  \ar[d]
\\
  C_{10}
  \ar[r]
&
  C_{11}
}
\]
under $A$.
We consider its pushout along $A \to B$, forming a cube
\[
\xymatrix@!C@C-0.2cm{
  C_{00}
  \ar[rr]
  \ar[dd]
  \ar[dr]
  \ar[rr]
  \fancypullback{[dd]}{[rrr]}
&&
  C_{01}
  \ar[dd]|!{[dl];[dr]}{\hole}
  \ar[dr]
&\\&
  D_{00}
  \ar[rr]
  \ar[dd]
&&
  D_{01}
  \ar[dd]
  \fancypullback{[ul]}{[ll]}
\\
  C_{10}
  \ar[rr]|!{[ur];[dr]}{\hole}
  \ar[dr]
&&
  C_{11}
  \ar[dr]
\\&
  D_{10}
  \ar[rr]
  \fancypullback{[uu]}{[ul]}
&&
  D_{11}
  \fancypullback{[uu]}{[ul]}
  \fancypullback{[ul]}{[ll]}
\rlap{.}}
\]
Our goal is to show that its front face is a pullback.
By \cref{pushout-mono-descent}, the maps from back to front are mono.
By pushout pasting, the faces from back to front are pushouts.
By \cref{pushout-mono-descent}, the left face is a pullback.
Hence, the front face is a pullback by descent.
\end{proof}

\begin{remark} \label{pushout-mono-almost-lex}
Let $\C$ be a finitely complete higher category with descent for pushouts.
Given a monomorphism $A \to B$ having pushouts, \cref{pushout-mono-preserves-pullbacks} shows that the pushout functor $\C \backslash A \to \C \backslash D$ is almost left exact, with only the terminal object not generally being preserved.
We can pass to slices to fix this: given a pushout square as in~\cref{pushout-mono:square}, the induced functor from the higher category of factorizations of $A \to C$ to the higher category of factorizations of $B \to D$ is left exact.
\end{remark}

For the following statement, note that the forgetful functor from a slice higher category creates truncation levels of maps.
That is, the truncation level of a map does not change if we view regard the map as living over different objects.

\begin{lemma} \label{pushout-truncated-helper-easy}
In a finitely complete higher category, let
\begin{equation} \label{pushout-truncated-helper-easy:pushout}
\begin{gathered}
\xymatrix{
  A
  \ar[r]
  \ar[d]
&
  C
  \ar[d]
\\
  B
  \ar[r]
&
  D
  \fancypullback{[u]}{[l]}
}
\end{gathered}
\end{equation}
be a pushout stable under pullback.
A map over $D$ is $n$-truncated exactly if its pullbacks over $B$ and $C$ are $n$-truncated.
\end{lemma}

\begin{proof}
We induct on $n$.
In the base case $n = -2$, this is functoriality of pushouts that exist.
The case $n \geq -1$ reduces to the case for $n - 1$ since diagonals of maps are preserved by pullbacks.
\end{proof}

The following lemma is a simple recognition criterion for truncation levels of pushouts.

\begin{lemma} \label{pushout-truncated-helper}
In a finitely complete higher category, let
\begin{equation} \label{pushout-truncated-helper:pushout}
\begin{gathered}
\xymatrix{
  A
  \ar[r]
  \ar[d]
&
  C
  \ar[d]
\\
  B
  \ar[r]
&
  D
  \fancypullback{[u]}{[l]}
}
\end{gathered}
\end{equation}
be a pushout stable under pullback.
Let $n \geq -1$.
Then $D$ is $n$-truncated exactly if the maps
\begin{enumerate}[label=(\roman*)]
\item \label{pushout-truncated-helper:B}
$B \times_D B \to B \times B$,
\item \label{pushout-truncated-helper:mixed}
$B \times_D C \to B \times C$,
\item \label{pushout-truncated-helper:C}
$C \times_D C \to C \times C$
\end{enumerate}
are $(n-1)$-truncated.
\end{lemma}

\begin{proof}
$D$ is $n$-truncated exactly if $D \to D \times D$ is $(n-1)$-truncated.
By pullback pasting, one sees that the listed maps are the pullbacks of the diagonal $D \to D \times D$ along the map $L \times R \to D \times D$ for all combinations of $L$ and $R$ being $B$ or $C$ (the map~\cref{pushout-truncated-helper:mixed} appears twice).
Thus, they are $(n-1)$-truncated if $D \to D \times D$ is $(n-1)$-truncated.
The reverse implication follows from a double invocation of \cref{pushout-truncated-helper-easy}.
\end{proof}

The \emph{join} $X \join Y$ of objects $X$ and $Y$, if it exists, is their pushout product, seen as maps $X \to 1$ and $Y \to 1$, \ie the pushout
\begin{equation} \label{join}
\begin{gathered}
\xymatrix{
  X \times Y
  \ar[r]
  \ar[d]
&
  X
  \ar[d]
\\
  Y
  \ar[r]
&
  X \join Y
  \fancypullback{[u]}{[l]}
\rlap{.}}
\end{gathered}
\end{equation}

\begin{lemma} \label{join-prop-level-descent}
In a finitely complete higher category with descent for pushouts, assume that the pushout~\cref{join} exists and is stable under pullback.
If $X$ and $Y$ are $(-1)$-truncated, then so is their join $X \join Y$.
\end{lemma}

\begin{proof}
By assumption, the left and top maps in~\cref{join} are mono.
Applying \cref{pushout-mono-descent} twice, the square is a pullback and the right and bottom maps are mono.
We check that $X \join Y$ is $(-1)$-truncated by instantiating \cref{pushout-truncated-helper} to the pushout~\cref{join}.
\begin{itemize}
\item 
Since $X \to X \join Y$ is mono, \cref{pushout-truncated-helper:B} is the diagonal $X \to X \times X$.
This is an equivalence since $X$ is $(-1)$-truncated.
\item
Since~\cref{join} is a pullback, \cref{pushout-truncated-helper:mixed} is an equivalence.
\item
The case for~\cref{pushout-truncated-helper:C} is analogous to the one for~\cref{pushout-truncated-helper:B}, using that $Y \to X \join Y$ is mono and $Y$ is $(-1)$-truncated.
\qedhere
\end{itemize}
\end{proof}

We have an analogous statement in the setting with classifiers, proved using \cref{pushout-mono} instead of \cref{pushout-mono-descent}.

\begin{lemma-qed}[Propositional Join] \label{join-prop-level}
In a locally Cartesian closed higher category, assume objects $X$ and $Y$ have classifiers also classifying the terminal object.
If $X$ and $Y$ are $(-1)$-truncated, then so is their join $X \join Y$ if it exists.
\end{lemma-qed}

For the application of \cref{join-prop-level} in the main body of the paper, we note that the assumptions on classifiers are satisfied in homotopy type theory as soon as we have univalence for propositions.

We finally come to the main result of this section.
We first state it using descent.
In this case, it has an easier proof.

\begin{proposition} \label{pushout-mono-n-truncated-descent}
Consider a finitely complete higher category with descent for pushouts.
Let $A \to B$ be a monomorphism having pushouts that are stable under pullback.
Consider a pushout
\begin{equation} \label{pushout-mono-n-truncated-descent:pushout}
\begin{gathered}
\xymatrix{
  A
  \ar[r]
  \ar[d]
&
  C
  \ar[d]
\\
  B
  \ar[r]
&
  D
  \fancypullback{[u]}{[l]}
\rlap{.}}
\end{gathered}
\end{equation}
If $B$ and $C$ are $n$-truncated for $n \geq 0$, then so is $D$.
This also holds for $n = -1$ if we have a map $B \times C \to A$.
\end{proposition}

\begin{proof}
We proceed by induction on $n$.
Of note, the inductive step will apply the claim to a slice, so the ambient category changes within the induction.
(In homotopy type theory, this simply corresponds to a use of the induction hypothesis in an extended context.)

In the base case $n = -1$, note that $A$ is $(-1)$-truncated since $A \to B$ is mono and $B$ is $(-1)$-truncated.
Since we have maps back and forth between $A$ and $B \times C$ and both objects are $(-1)$-truncated, these maps are invertible.
It follows that the span in~\cref{pushout-mono-n-truncated-descent:pushout} is a product span, and the claim reduces to \cref{join-prop-level-descent}.

From now on, assume $n \geq 0$.
By \cref{pushout-mono-descent}, the map $C \to D$ is mono and~\cref{pushout-mono-n-truncated-descent:pushout} is a pullback.
We check that $D$ is $n$-truncated by instantiating \cref{pushout-truncated-helper} to the pushout~\cref{pushout-mono-n-truncated-descent:pushout}.
\begin{itemize}
\item
Since $C \to D$ is mono, the map~\cref{pushout-truncated-helper:C} is the diagonal $C \to C \times C$.
This is $(n-1)$-truncated since $C$ is $n$-truncated.
\item
Since~\cref{pushout-mono-n-truncated-descent:pushout} is a pullback, the map~\cref{pushout-truncated-helper:mixed} is $A \to B \times C$.
This factors as
\[
\xymatrix{
  A
  \ar[r]
&
  A \times C
  \ar[r]
&
  B \times C
\rlap{.}}
\]
The first factor is a pullback of the diagonal $C \to C \times C$, which is $(n-1)$-truncated since $C$ $n$-truncated.
The second factor is a pullback of $A \to B$, which is $(n-1)$-truncated since $n \geq 0$ and $A \to B$ is mono.
\end{itemize}
The remaining case is the map~\cref{pushout-truncated-helper:B}, \ie showing that $B \times_D B \to B \times B$ is $(n-1)$-truncated.
Consider the diagram producing the diagonal of $A$ over $C$:
\begin{equation} \label{pushout-mono-n-truncated-descent:P}
\begin{gathered}
\xymatrix{
  A
  \ar@{.>}[dr]
  \ar@/_1em/[ddr]_{\id}
  \ar@/^1em/[drr]^{\id}
\\
&
  P
  \ar[r]^{p_1}
  \ar[d]_{p_2}
  \fancypullback{[d]}{[r]}
&
  A
  \ar[d]
\\
&
  A
  \ar[r]
&
  C
\rlap{.}}
\end{gathered}
\end{equation}
We see this as a square in the coslice under $A$.
We take its pushout along $A \to B$:
\begin{equation} \label{pushout-mono-n-truncated-descent:Q}
\begin{gathered}
\xymatrix{
  B
  \ar[dr]
  \ar@/_1em/[ddr]_{\id}
  \ar@/^1em/[drr]^{\id}
\\
&
  Q
  \ar[r]^{q_1}
  \ar[d]_{q_2}
&
  B
  \ar[d]
\\
&
  B
  \ar[r]
&
  D
\rlap{.}}
\end{gathered}
\end{equation}
We will independently prove the following:
\begin{enumerate}[label=(\roman*)]
\item \label{pushout-mono-n-truncated-descent:Q-is-pullback} the inner square in~\cref{pushout-mono-n-truncated-descent:Q} is a pullback,
\item \label{pushout-mono-n-truncated-descent:Q-truncated} the map $\angles{q_1, q_2} \co Q \to B \times B$ is $(n-1)$-truncated.
\end{enumerate}
Together, this implies the goal.

Let us prove~\cref{pushout-mono-n-truncated-descent:Q-is-pullback}.
By pushout pasting, the squares
\begin{equation} \label{pushout-mono-n-truncated-descent:pushout-pasting}
\begin{aligned}
\xymatrix{
  P
  \ar[r]^{p_1}
  \ar[d]
&
  A
  \ar[d]
\\
  Q
  \ar[r]^{q_1}
&
  \fancypullback{[u]}{[l]}
  B
}
&\qquad\qquad&
\xymatrix{
  P
  \ar[r]^{p_2}
  \ar[d]
&
  A
  \ar[d]
\\
  Q
  \ar[r]^{q_2}
&
  \fancypullback{[u]}{[l]}
  B
}
\end{aligned}
\end{equation}
are pushouts.
By~\cref{pushout-mono-descent}, the map $P \to Q$ is mono.
Again by~\cref{pushout-mono-descent}, the squares~\cref{pushout-mono-n-truncated-descent:pushout-pasting} are pullbacks.
Applying descent in the cube connecting the inner squares in~\cref{pushout-mono-n-truncated-descent:P,pushout-mono-n-truncated-descent:Q}, we see that the inner square in~\cref{pushout-mono-n-truncated-descent:Q} is a pullback.

Let us prove~\cref{pushout-mono-n-truncated-descent:Q-truncated}.
Focus on the pushout square
\begin{equation} \label{pushout-mono-n-truncated-descent:0}
\begin{gathered}
\xymatrix{
  A
  \ar[r]
  \ar[d]
&
  P
  \ar[d]
\\
  B
  \ar[r]
&
  \fancypullback{[u]}{[l]}
  Q
\rlap{.}}
\end{gathered}
\end{equation}
It lives in the slice over $B \times B$ via $\angles{q_1, q_2}$.
The induced map $P \to B \times B$ rewrites as
\[
\xymatrix@C+0.2cm{
  P
  \ar[r]^-{\angles{p_1, p_2}}
&
  A \times A
  \ar[r]
&
  B \times B
\rlap{.}}
\]
As in the proof of \cref{pushout-truncated-helper}, the first factor is a pullback of the diagonal $C \to C \times C$, hence $(n-1)$-truncated.
The second factor is mono, hence $(n-1)$-truncated since $n \geq 0$.
It follows that $P$ is $(n-1)$-truncated over $B \times B$.
Note that $B$ is $(n-1)$-truncated over $B \times B$ by assumption.

We apply the induction hypothesis to the pushout~\cref{pushout-mono-n-truncated-descent:0} in the slice over $B \times B$.
Note that the assumptions on descent and the mono $A \to B$ all descend to the slice.
We have just shown that the bottom left and top right objects of~\cref{pushout-mono-n-truncated-descent:0} are $(n-1)$-truncated over $B \times B$.
For the case $n-1 = -1$, we note that we have maps
\[
\xymatrix{
  B \times_{B \times B} P
  \ar[r]
&
  B \times_{B \times B} (A \times A)
&
  A
  \ar[l]_-{\simeq}
}
\]
over $B \times B$ where the second map inverts as $A \to B$ is mono.
This finishes the proof of~\cref{pushout-mono-n-truncated-descent:Q-truncated}.
\end{proof}

We now state the above result in terms of classifiers.
Since we only assume classifiers for $(n-1)$-truncated maps, the previous proof requires modification.

\begin{proposition}[$n$-truncated Pushout Of Mono] \label{pushout-mono-n-truncated}
In a locally Cartesian closed higher category, let $m \co A \to B$ be a monomorphism having pushouts.
Consider a pushout
\begin{equation} \label{pushout-mono-n-truncated:pushout}
\begin{gathered}
\xymatrix{
  A
  \ar[r]
  \ar[d]_{m}
&
  C
  \ar[d]
\\
  B
  \ar[r]
&
  D
  \fancypullback{[u]}{[l]}
\rlap{.}}
\end{gathered}
\end{equation}
Let $n \geq 0$.
Assume a classifier for any finite collection of $(n-1)$-truncated maps.
If $B$ and $C$ are $n$-truncated, then so is $D$.
This also holds for $n = -1$ if we have a map $B \times C \to A$ and a classifier for any finite collection of monomorphisms.
\end{proposition}

\begin{proof}
The proof proceeds as for \cref{pushout-mono-n-truncated-descent}, using \cref{join-prop-level} instead of \cref{join-prop-level-descent} in the base case $n = -1$.

In the situation for $n \geq 0$, we use \cref{pushout-mono} instead of \cref{pushout-mono-descent} to deduce that $C \to D$ is mono and~\cref{pushout-mono-n-truncated:pushout} is a pullback.

A divergence occurs when checking the subgoals~\cref{pushout-mono-n-truncated-descent:Q-is-pullback,pushout-mono-n-truncated-descent:Q-truncated}.
The previous descent argument for~\cref{pushout-mono-n-truncated-descent:Q-is-pullback} no longer works because the relevant maps in the cube considered are not $(n-1)$-truncated.
Instead, we are forced to prove~\cref{pushout-mono-n-truncated-descent:Q-is-pullback} in a way that will depend on~\cref{pushout-mono-n-truncated-descent:Q-truncated}.
The proof of~\cref{pushout-mono-n-truncated-descent:Q-truncated} proceeds as before, noting that the assumptions on classifiers descend to the recursive case.

Let us now prove~\cref{pushout-mono-n-truncated-descent:Q-is-pullback}.
To start, we use \cref{pushout-mono} instead of \cref{pushout-mono-descent} to derive the assertions about the squares~\cref{pushout-mono-n-truncated-descent:pushout-pasting}.
Consider the cube
\begin{equation} \label{pushout-mono-n-truncated:real-cube}
\begin{gathered}
\xymatrix@!C@C-0.3cm{
  P
  \ar[rr]^{p_1}
  \ar[dd]_{\angles{m p_1, p_2}}
  \ar[dr]
  \ar[rr]
  \fancypullback{[dd]}{[rrr]}
  \fancypullback{[dd]}{[dr]}
&&
  A
  \ar[dd]|!{[dl];[dr]}{\hole}
  \ar[dr]
&\\&
  Q
  \ar[rr]^(0.3){q_1}
  \ar[dd]_(0.3){\angles{q_1, q_2}}
&&
  B
  \ar[dd]
  \fancypullback{[ul]}{[ll]}
\\
  B \times A
  \ar[rr]|!{[ur];[dr]}{\hole}
  \ar[dr]
&&
  B \times C
  \ar[dr]
\\&
  B \times B
  \ar[rr]
&&
  B \times D
  \fancypullback{[ul]}{[ll]}
\rlap{.}}
\end{gathered}
\end{equation}
The bottom face is the pullback of~\cref{pushout-mono-n-truncated:pushout} along $B \to 1$, hence a pushout by local Cartesian closure.
The top face is one of the pushouts~\cref{pushout-mono-n-truncated-descent:pushout-pasting}.
To see that the back and left faces are pullbacks, pullback paste in the diagrams
\begin{align*}
\xymatrix{
  P
  \ar[r]^-{p_1}
  \ar[d]_{\angles{m p_1, p_2}}
&
  A
  \ar[d]
\\
  B \times A
  \ar[r]
  \ar[d]_{\pi_2}
  \fancypullback{[d]}{[r]}
&
  B \times C
  \ar[d]_{\pi_2}
\\
  A
  \ar[r]
&
  C
\rlap{,}}
&&
\xymatrix{
  P
  \ar[r]
  \ar[d]_{\angles{m p_1, p_2}}
&
  Q
  \ar[d]_{\angles{q_1, q_2}}
\\
  B \times A
  \ar[r]
  \ar[d]_{\pi_2}
  \fancypullback{[d]}{[r]}
&
  B \times B
  \ar[d]_{\pi_2}
\\
  A
  \ar[r]
&
  B
\rlap{;}}
\end{align*}
here, the left composite square is a pullback by construction~\cref{pushout-mono-n-truncated-descent:P} and the right composite square is one of the pullbacks~\cref{pushout-mono-n-truncated-descent:pushout-pasting}.

We already know that $A \to B \times C$ is $(n-1)$-truncated.
The map $Q \to B \times B$ is $(n-1)$-truncated by~\cref{pushout-mono-n-truncated-descent:Q-truncated}.
By assumption, they have a common classifier.
We are thus in the position to apply \cref{descent} to the cube~\cref{pushout-mono-n-truncated:real-cube} and deduce that its front face is a pullback.
Pasting the pullbacks
\[
\xymatrix{
  Q
  \ar[r]^{q_1}
  \ar[d]_{\angles{q_1, q_2}}
  \fancypullback{[d]}{[r]}
&
  B
  \ar[d]
\\
  B \times B
  \ar[r]
  \ar[d]_{\pi_2}
  \fancypullback{[d]}{[r]}
&
  B \times D
  \ar[d]_{\pi_2}
\\
  B
  \ar[r]
&
  D
\rlap{,}}
\]
we get \cref{pushout-mono-n-truncated-descent:Q-is-pullback}.
\end{proof}

\begin{remark} \label{pushout-mono-n-truncated-descent-restricted}
Following the idea of the proof of \cref{pushout-mono-n-truncated}, one may strengthen also the statement of \cref{pushout-mono-n-truncated-descent} by restricting descent to $(n-1)$-descent (or $(-1)$-descent for $n = -1$).
Here, \emph{$n$-descent} in a finitely complete higher category $\C$ refers to the descent-like notion obtained by instead considering the functor from $\C^\op$ to higher categories that sends $Z$ to the full higher subcategory of $C/Z$ on $n$-truncated objects.
\end{remark}

\begin{remark}
It is possible to state the assumption on classifiers in \cref{pushout-mono-n-truncated} in a form closer to current systems of homotopy type theory.
This involves introducing a notion of universe, \ie fixing a (not necessarily univalent) map $V \to U$ whose classified maps are closed under composition and formation of diagonals and whose classified objects in any slice are closed under pushouts.
One then requires $A, B, C$ to be classified by $V \to U$ and that the collection of those maps classified by $V \to U$ that are $(n-1)$-truncated (or $(-1)$-truncated for $n = -1$) also admits a (univalent) classifier in our sense.

Note that we are speaking here about classification in the sense of higher categories, not strict classification by some universe in a model of type theory.
This is the essential difference to the topic of the main body of the article where this statement is used.
\end{remark}

For the application of \cref{pushout-mono-n-truncated} in the main body of the paper, we note that the assumptions on classifiers are satisfied in homotopy type theory as soon as we have univalence for propositions.

\section{Partially Propositional Indexed W-types}
\label{appendix-on-W-type-stuff}

In this section, we develop the theory of partially propositional indexed W-types and characterize their equality.
As a warm up, we recall indexed containers, the associated notion of indexed-type, and the encode-decode method used to characterize their equality.
The results developed here will be used in \cref{sec:general} in the definition of an $n$-truncated universe of $n$-types that is univalent for $(n-1)$-types.

We work informally in the language of homotopy type theory.
In particular, we have access to function extensionality.
We will be explicit about universes when they are needed and what kind of univalence we require of them.

No result in this section depends on judgmental $\beta$-equality for higher inductive types, even for point constructors.
The $\beta$-law as an internal equality will suffice.
This applies to pushouts, the higher inductive families that will constitute partially propositional W-types, and even ordinary indexed W-types.

\subsection{Indexed Containers}
\label{indexed-container}

Given a type $I$, an \emph{$I$-indexed container} $C$ is a pair $C = (S, \Pos)$ of type families as follows:
\begin{itemize}
\item
given $i : I$, we have a type $S(i)$,
\item
given $i, j : I$, and $s : S(i)$, we have a type $\Pos(s, j)$.%
\footnote{This is a polynomial functor $I \leftarrow E \to B \to I$ with Reedy fibrant specifying data.
The latter means that $B \to I$ and $E \to B \times I$ are fibrations.}
\end{itemize}
Its \emph{extension} is the endofunctor on families over $I$ that sends a family $X$ to the family $\cExt(X)$ given by
\[
\cExt(X)(i) = \sm{s:S(i)} \prd{j:I} \Pos(i, s, j) \to X(j)
.\]

\subsection{Indexed W-types}
\label{indexed-W-types}

Fix an $I$-indexed container $C$ as above.
As in~\cite{sojakova-et-al:homotopy-W-types,kaposi-kovacs:hits-syntax} (the former applying only to the non-indexed case), one has notions of \emph{algebras}, \emph{algebra morphisms}, \emph{algebra fibrations}, and \emph{algebra fibration sections} for $C$.
Here and in the following, we omit the prefix ``homotopy'' for all relevant notions, this being the default meaning for us.

\begin{definition}[Indexed W-type]
A \emph{W-type} for the indexed container $C$ is an initial $C$-algebra, meaning the type of algebra morphisms to any algebra is contractible,
\end{definition}

We denote a given, substitutionally stable choice of such an object by $(W_C, \sup)$, although denoting the whole algebra by $W_C$.
One may characterize $W_C$ via elimination: any algebra fibration over $W_C$ has an algebra section.
This is known as \emph{induction}.
Note that the $\beta$-law for the eliminator holds only up to identity type.

In a framework with inductive families (with or without judgmental $\beta$-law), one may implement $W_C$ as an inductive family with constructor $\sup(s, t) : W_C(i)$ for $i : I$, $s : S(i)$, and $t : \prd{j : I} \Pos(s, j) \to W_C(j)$; when applying $t$, we leave its first argument $j$ implicit.
To keep the presentation readable, we will informally use $W_C$ as if it was given as such an inductive family and use pattern-matching-style notation for induction; we leave the reduction to the eliminator (including the definition of the algebra fibration corresponding to each use of induction) to the reader.

\subsection{Encode-Decode Method}
\label{encode-decode}

The encode-decode method~\cite{licata-shulman:encode-decode} is a general method for characterizing equality in a (higher) inductive type (or family) $T$.
In our view, it decomposes into the following three steps (here for just a single type $T$).
\begin{enumerate}
\item \label{encode-decode:code}
Define a binary relation $\Eq_T$ of \emph{equality codes} on $T$.
This uses double induction on $T$ (with ultimate target a universe) and makes use of univalence if there are any path constructors.
\item \label{encode-decode:encode}
Define an ``encoding'' function
\begin{equation} \label{encode}
\begin{gathered}
\xymatrix@C+0.8cm{
  T(x_0, x_1)
  \ar[r]^-{\encode_{x_0, x_1}}
&
  \Eq_T(x_0, x_1)
\rlap{.}}
\end{gathered}
\end{equation}
This uses induction over the given equality and single induction on $T$.
\item \label{encode-decode:encode-fiber}
Prove $\encode^{-1}(c)$ for each code $c : \Eq_T(x_0, x_1)$.
This is a pair $(p, q)$ where $p : T(x_0, x_1)$ and $q : \encode(p) = c$.
This uses double induction on $T$ as in step~\cref{encode-decode:code}.
\end{enumerate}
Step~\cref{encode-decode:encode-fiber} makes~\cref{encode} into a retraction.
Summing over $x_1 : T$, the source of the retraction becomes contractible.
But any retraction with contractible source is an equivalence.
By a standard lemma about fiberwise equivalences~\cite[Thm~4.7.7]{HoTTBook}, it follows that the original map~\cref{encode} is an equivalence.

\subsection{Equality in Indexed W-types}
\label{indexed-w-types-equality}

We now apply the encode-decode method from~\cref{encode-decode} to characterize equality in $W_C$.
Although we did not find this in the literature, we believe it to be folklore.
We have a choice between two viable options:
\begin{itemize}
\item
characterize equality in each fiber $W_C(i)$ for $i : I$,
\item
characterize dependent equalitiy in $W_C$ over a given equality in $I$.
\end{itemize}
We go for the second one.

For step~\cref{encode-decode:code}, we define a type $\Eq(p_i, x_0, x_1)$ of \emph{equality codes} between $x_0 : W_c(i_0)$ and $x_1 : W_C(i_1)$ over $p_i : I(i_0, i_1)$.
We proceed by double induction on $x_0$ and $x_1$ into a large enough universe.
For $x_0 \equiv \sup(s_0, t_0)$ and $x_1 \equiv \sup(s_1, t_1)$, we take $\Eq(p_i, x_0, x_1)$ equal to the type of pairs $(p_s, c_t)$ where:
\begin{itemize}
\item $p_s$ is a dependent equality in $S$ over $p_i$ between $s_0$ and $s_1$,
\item $c_t$ is a dependent function, sending $j : I$\footnote{Instead of $j : I$, we could also quantify over $j_0, j_1 : I$ with $p_j : I(j_0, j_1)$. Perhaps this would be more consistent.}, $m_0 : \Pos(s_0, j)$, $m_1 : \Pos(s_1, j)$, and a dependent equality $p_m$ over $p_i$ and $p_s$ between $m_0$ and $m_1$ to
\[
c_t(p_m) : \Eq(\refl_j, t_0(m_0), t_1(m_1))
.\]
\end{itemize}

Steps~\cref{encode-decode:encode,encode-decode:encode-fiber} will establish the following statement.

\begin{proposition} \label{indexed-W-equality-equiv}
Given $p : I(i_0, i_1)$ with $x_0 : W_C(i_0)$ and $x_1 : W_C(i_1)$, we have
\[
(x_0 =_{W_C(p)} x_1) \simeq \Eq(p, x_0, x_1)
.\]
\end{proposition}

\begin{proof}
For step~\cref{encode-decode:encode}, we first define $\encode'(x) : \Eq(\refl_i, x, x)$ for $i : I$ and $x : W_C(i)$.
We proceed by induction on $x$, letting $x \equiv \sup(s, t)$.
We let $\encode'(x)$ be $(\refl_s, p_t)$, transported along the $\beta$-equality for $\Eq$.
Here, $p_t$ is defined by equality induction from $p_t(j, \refl_m) \defeq \encode'(t(m))$.
From $\encode'$, we obtain
\[
\xymatrix@C+1.2cm{
  x_0 =_{W_C(p_i)} x_1
  \ar[r]^-{\encode_{p_i, x_0, x_1}}
&
  \Eq(p_i, x_0, x_1)
}
\]
by equality induction first on $p_i$ and then the equality between $x_0$ and $x_1$ in the same fiber.

For step~\cref{encode-decode:encode-fiber}, we prove $\encode_{p_i, x_0, x_1}^{-1}(c)$ for $c : \Eq(p_i, x_0, x_1)$.
Inducting on $p_i$, we may suppose $i \defeq i_0 \equiv i_1$ and $p_i \equiv \refl$.
We induct on $x_0$ and $x_1$, letting $x_0 \equiv \sup(s_0, t_0)$ and $x_1 \equiv \sup(s_1, t_1)$.
Using the $\beta$-equality for $\Eq$, we can reduce to the case where $c$ is the transport of a pair $(p_s, c_t)$ as in step~\cref{encode-decode:code}.
Inducting on $p_s$, we may suppose $s \defeq s_0 \equiv s_1$ and $p_s \equiv \refl$.

Given $j : I$, note that the canonical map
\[
\Pos(s, j) \to \textstyle{\sm{m_0, m_1 : \Pos(s, j)} m_0 = m_1}
\]
is an equivalence.
Using function extensionality and equality induction, we may thus suppose that $c_t$ is obtained using equality induction from
\[
c_t'(m) \defeq c_t(\refl_m) : \Eq(\refl_j, t_0(m), t_1(m))
\]
for $j : I$ and $m : \Pos(s, j)$ in the same fashion that $p_t$ is defined in the definition of $\encode'$.
By induction hypothesis, we have $\encode_{\refl_j, t_0(m), t_1(m)}^{-1}(c_t'(m))$.
Using function extensionality and equality induction, we may thus suppose that
\begin{align*}
c_t'
&\equiv
\lam{j}{m} \encode_{\refl_j, t_0(m), t_1(m)}(q_m(m))
\end{align*}
for some $q_m(m) : t_0(m) =_{W_C(j)} t_1(m)$ depending on $j : I$ and $m : \Pos(s, j)$.
Again using function extensionality and equality induction, we may suppose that $t \defeq t_0 \equiv t_1$ and $q_m \equiv \lam{j}{m} \refl_{t(m)}$.

Now we have $c_t' \equiv \lam{j}{m} \encode'(t(m))$, and thus
\begin{align*}
c
&\equiv
(\refl_s, c_t)
\\&\equiv
\encode'(\sup(s, t))
\\&\equiv
\encode_{\refl_i, \sup(s, t), \sup(s, t)}(\refl)
.\end{align*}
This shows $\encode_{p_i,x_0,x_1}^{-1}(c)$.
\end{proof}

\subsection{Partially Propositional Indexed W-types}
\label{partially-propositional-indexed-W-types}

Fix an $I$-indexed container $C$.
In addition, fix a family $P$ of propositions over $I$.

We wish to mix the inductive construction of the indexed W-type $W_C$ with the propositional truncation, applied only over indices $i : I$ with $P(i)$.
Note that indexed W-types are not merely fiberwise W-types.
Thus, it would be wrong to take $W_C$ and simply conditionally truncate each fiber,
Rather, we want to have the truncation interleaved within the inductive definition.

A type $A$ is \emph{$Q$\hyp{}propositional} if $A$ is propositional assuming $Q$, i.e., if $Q \to \isprop(A)$.
A family $X$ over $I$ is \emph{$P$\hyp{}propositional} if $X(i)$ is $P(i)$-propositional for $i : I$.
We use analogous terminology with being contractible.

A $C$-algebra is \emph{$P$\hyp{}propositional} if the underlying family is $P$-propositional.

\begin{definition}[Partially propositional indexed W-type] \label{partially-propositional-indexed-W-type}
A \emph{partially propositional W-type} for the propositional family $P$ and the indexed container $C$ is an initial $P$\hyp{}propositional $C$\hyp{}algebra.
\end{definition}

We denote a given, substitutionally stable choice of such an object by $W_C^P$.
We also call this the \emph{$P$\hyp{}propositional indexed W-type} generated by by $C$.
As before, one may characterize $W_{C,P}$ via elimination (\ie induction): any $P$\hyp{}propositional algebra fibration over $W_C^P$ has an algebra section.
Here, an algebra fibration over a $P$\hyp{}propositional algebra is \emph{$P$\hyp{}propositional} if the total space algebra is $P$\hyp{}propositional.

Equivalently, $W_C^P$ is the higher inductive family~\cite{lumsdaine-shulman:hits-semantics,kaposi-kovacs:hits-syntax} with point constructor $\sup$ as for $W_C$ and path constructor
\[
\trunc(c, x, y) : W_C^P(i)(x, y)
\]
for $i : I$ with $c : P(i)$ and $x, y : W_C^P(i)$.
As before, we do not assume that $\beta$-equality is judgmental.
Again, to keep the presentation readable, we will informally use induction on $W_C^P$ using pattern-matching-style notation, leaving the compilation to elimination with respect to $P$\hyp{}propositional algebra fibrations to the reader.
The reasoning principle is the same as for $W_C$, only that we must show the target's fiber over $i : I$ is propositional if $P(i)$.

Higher inductive families have been justified in the simplicial set model~\cite{lumsdaine-shulman:hits-semantics} and various cubical sets models~\cite{cavallo:hits,coquand-et-al:hits} of homotopy type theory.
Thus, these models support partially propositional W-types (even with judgmental $\beta$-reduction).

\subsection{Equality in Partially Propositional Indexed W-types}
\label{partially-propositional-indexed-W-types-equality}

We adapt the encode-decode method from~\cref{encode-decode} to characterize equality in $W_C^P$.
Large parts of the construction will follow the technical development in \cref{indexed-w-types-equality}, so the focus will lie on the new aspect of $P$-propositionality.

In step~\cref{encode-decode:code}, we need to define codes for equality. 
This relies on the following lemma.

\begin{lemma} \label{partial-contractiblity-lemma}
Let $\UU$ be a universe univalent for contractible types.
Let $Q$ be proposition.
Then the subtype of $\UU$ of $Q$-contractible types is $Q$-contractible.
\end{lemma}

\begin{proof}
If $Q$ holds, $Q$-contractibility means contractibility.
So the goal becomes: if $Q$ holds, the subtype of $\UU$ of contractible types is contractible.
But the latter is always contractible, using univalence.
\end{proof}

Given $p_i : I(i_0, i_1)$ with $x_0 : W_c^P(i_0)$ and $x_1 : W_C^P(i_1)$, we define simultaneously:
\begin{itemize}
\item a type $\Eq(p_i, x_0, x_1)$ of \emph{equality codes},
\item a witness that $\Eq(p_i, x_0, x_1)$ is $P(i_0)$-contractible.%
\footnote{Note that $P(i_0) \leftrightarrow P(i_1)$ given $p_i : I(i_0, i_1)$.}
\end{itemize}
We proceed by double elimination on $x_0$ and $x_1$ with target over $i : I$ the type of $P(i)$-contractible types in a large enough universe $\UU$ univalent for contractible types.
This is a $P$-propositional algebra by \cref{partial-contractiblity-lemma}.
For $x_0 \equiv \sup(s_0, t_0)$ and $x_1 \equiv \sup(s_1, t_1)$, we take the join
\begin{equation} \label{partially-propositional-Eq}
\Eq(p_i, x_0, x_1) = P(i_0) \join \Eq'(p_i, (s_0, t_0), (s_1, t_1))
.\end{equation}
Here, $\Eq'(p_i, (s_0, t_0), (s_1, t_1))$ codes structural equality of the top-level constructor applications.
It is defined as the type of pairs $(p_s, c_t)$ as in the construction of $\Eq$ in \cref{indexed-w-types-equality}, with $c_t$ ultimately valued in $\Eq$.
The join~\cref{partially-propositional-Eq} may also be understood as a $P(i_0)$-partial contractible truncation.
In particular, it is contractible if $P(i_0)$.

\begin{proposition} \label{partially-propositional-indexed-W-equality-equiv}
Given $p : I(i_0, i_1)$ with $x_0 : W_C^P(i_0)$ and $x_1 : W_C^P(i_1)$, we have
\[
(x_0 =_{W_C^P(p)} x_1) \simeq \Eq(p, x_0, x_1)
.\]
\end{proposition}

\begin{proof}
It remains to adapt steps~\cref{encode-decode:encode,encode-decode:encode-fiber} of the encode-decode method.

For step~\cref{encode-decode:encode}, we first define $\encode'(x) : \Eq(\refl_i, x, x)$ for $i : I$ and $x : W_C(i)$.
Note that the goal is contractible if $P(i)$.
We may thus induct on $x$.
For $x \equiv \sup(s, t)$, we define $\encode'(x) = \inr(\ldots)$ where the omitted expression is as in the constructor case for $\encode'$ in \cref{indexed-W-equality-equiv}.
From $\encode'$, obtain 
\[
\xymatrix@C+1.2cm{
  x_0 =_{W_C^P(p_i)} x_1
  \ar[r]^-{\encode_{p_i, x_0, x_1}}
&
  \Eq(x_0, x_1)
}
\]
as in \cref{indexed-W-equality-equiv}.

For step~\cref{encode-decode:encode-fiber}, we prove $\encode_{p_i, x_0, x_1}^{-1}(c)$ for $c : \Eq(p_i, x_0, x_1)$.
Inducting on $p_i$, we may suppose $i \defeq i_0 \equiv i_1$ and $p_i \equiv \refl$.
Note that the goal becomes contractible if $P(i)$: source and target of $\encode_{\refl_i, x_0, x_1}$ become contractible.
Thus, we may induct on $x_0$ and $x_1$, letting $x_0 \equiv \sup(s_0, t_0)$ and $x_1 \equiv \sup(s_1, t_1)$.

Using the $\beta$-equality for $\Eq$, we can reduce to the case where $c$ is the transport of an element of the join~\cref{partially-propositional-Eq}.
We eliminate over this element of the join.
In the case for $\inl$ or $\glue$, we have $P(i)$, so are done as the goal is contractible.
In the case for $\inr$, we have that $c$ is the transport of $\inr(p_s, c_t)$ where $(p_s, c_t)$ is as in \cref{indexed-W-equality-equiv}.
The rest of the proof follows that proof.
\end{proof}

\begin{remark} \label{partially-n-truncated-indexed-W-type}
An evident generaliation of partially propositional indexed W-types is \emph{partially $n$-truncated indexed W-types} for internal or external $n \geq -2$.
We expect that the encode-decode method can also be adapted to characterize equality in partially $n$-truncated indexed W-types.
Instead of using the join as in~\cref{partially-propositional-indexed-W-equality-equiv}, one defines $\Eq(p_i, x_0, x_1)$ as the pushout
\[
\xymatrix@C-1.5cm@R-0.5cm{
  P(i_0) \times \Eq'(p_i, (s_0, t_0), (s_1, t_1))
  \ar[dr]
  \ar[dd]
\\&
  P(i_0) \times \vertss{\Eq'(p_i, (s_0, t_0), (s_1, t_1))}_{n-1}
  \ar[dd]
\\
  \Eq'(p_i, (s_0, t_0), (s_1, t_1))
  \ar[dr]
\\&
  \Eq(p_i, x_0, x_1)
  \fancypullback{[uu]}{[ul(0.6)]}
\rlap{.}}
\]
This is motivated by the equality in an $n$-truncation being the $(n-1)$-truncation of the original equality.
For $n = -1$, this definition reduces to~\cref{partially-propositional-indexed-W-equality-equiv}.
For $n = -2$, the definition still makes sense if we take the $(-3)$-truncation to mean $(-2)$-truncation.

Note that there is an induced map
\begin{equation} \label{partially-n-truncated-indexed-W-type:0}
\begin{gathered}
\Eq(p_i, x_0, x_1) \to \vertss{\Eq'(p_i, (s_0, t_0), (s_1, t_1))}_{n-1}
.\end{gathered}
\end{equation}
This is the pushout product of $P(i_0) \to 1$ with the map
\begin{equation} \label{partially-n-truncated-indexed-W-type:1}
\begin{gathered}
\Eq'(p_i, (s_0, t_0), (s_1, t_1)) \to \vertss{\Eq'(p_i, (s_0, t_0), (s_1, t_1))}_{n-1}
\end{gathered}
\end{equation}
Thus, the fiber of~\cref{partially-n-truncated-indexed-W-type:0} over an element $e$ is given by the join of $P(i_0)$ and the fiber of~\cref{partially-n-truncated-indexed-W-type:1} over $e$.
Thus, the join still plays a role, even though it is hidden in a different slice.
\end{remark}
}

\end{document}